\documentclass[11pt]{article}
\usepackage{fullpage}
\usepackage{times}
\usepackage{amsmath,amsthm}
\usepackage{graphicx}
\usepackage{xcolor}
\usepackage{comment}
\usepackage{url}
\usepackage{cite}
\newtheorem{theorem}{Theorem}[section]

\newtheorem{claim}[theorem]{Claim}
\newtheorem{lemma}[theorem]{Lemma}

\newtheorem{definition}[theorem]{Definition}

\newcommand{\calV}{{\cal{V}}}
\newcommand{\calE}{{\cal{E}}}
\newcommand{\calU}{{\cal{U}}}

\newcommand{\Tdiff}{T_{\textrm{\sc diff}}}
\newcommand{\Tmin}{T_{\textrm{\sc min}}}
\newcommand{\Tmax}{T_{\textrm{\sc max}}}
\newcommand{\Dmax}{\Delta_{\textrm{\sc max}}}
\newcommand{\RNmax}{{\textrm{\sc Rn}}_{\textrm{\sc max}}}
\newcommand{\ta}{t_{1}}

\newcommand{\tp}{t'}

\newcommand{\hear}{\mbox{Hear}}
\newcommand{\pre}{\mbox{Prev}}

\newcommand{\Exp} {\textrm{\sc E}}

\newcommand{\outset}{\textrm{\sc Out}}
\newcommand{\unset}{\textrm{\sc Un}}
\newcommand{\inset}{\textrm{\sc In}}

\newcommand{\Vunset}{\textrm{\sc V(Un)}}

\newcommand{\euu}{\textrm{\sc e}}

\newcommand{\NR}{\textrm{\sc Nr}}
\newcommand{\R}{\textrm{\sc R}}
\newcommand{\NL}{\textrm{\sc Nl}}
\newcommand{\Lh}{\hat{L}}

\newcommand{\inNin}{{\textrm{\sc In}^\textrm{\sc In}}}
\newcommand{\unNin}{{\textrm{\sc Un}^\textrm{\sc In}}}

\newcommand{\inNnin}{{\textrm{\sc In}^{\overline{\textrm{\sc In}}}}}
\newcommand{\unNnin}{{\textrm{\sc Un}^{\overline{\textrm{\sc In}}}}}
\newcommand{\outNnin}{{\textrm{\sc Out}^{\overline{\textrm{\sc In}}}}}

\newcommand{\unNninGT}{{\textrm{\sc Un}^{\overline{\textrm{\sc In}},>}}}

\newcommand{\barN}{{\overline{N}}}

\title{Singing a Maximal Independent Set}

\author{
  Sandy Irani\thanks{Department of Computer Science, University of California, Irvine.}
  \and
  Michael Luby\thanks{BitRipple Inc.}
}

\date{}

\begin{document}
\maketitle

\begin{abstract}
We introduce a broadcast model called the {\bf singing model}, where agents are oblivious of the size and structure of the communication network, even their immediate neighborhood. Agents can {\bf sing multiple notes} which are heard by their neighbors. The model is a generalization of the beeping model \cite{alon2011, CK10}, where agents can only emit sound at a single frequency. We give a simple and natural protocol where agents compete with their neighbors and their strength is reflected in the number of notes they sing. It converges in $O(\log(n))$ time with high probability, where $n$ is the number of agents in the network. The protocol works in an asynchronous model where rounds vary in length and have different start times. It works with completely dynamic networks where agents can be faulty. The protocol is the first to converge to an MIS in logarithmic time for dynamic networks in a network oblivious model.
\end{abstract}


\section{Introduction}



We study a model, motivated by natural processes,  where agents are connected by a communication network and must collectively solve a given task. The edges of the network indicate pairs of agents that can communicate directly. An important and natural constraint arising in many settings is that the agents do not have any knowledge of the network structure or size, even their immediate neighborhood. 
We introduce the {\bf singing model}, where
each agent can broadcast  ({\bf sing}) signals at  different frequencies, which will be heard by its neighbors in the network.  
Singing can be viewed as concurrently emitting one or more sound frequencies,
where each sound frequency can be thought of as a {\bf note}. The set of possible notes that can be sung is globally known, and each agent can either sing one or more of these notes or be silent during a round. 
An agent can also hear different notes and can detect the lack of any note transmitted by its neighbors.
However, an agent cannot determine the number of neighbors transmitting a given note.

There are many  compelling  phenomena, both natural and engineered, that fit into this model,
such as signaling networks between and within cells,  a population of communicating drones, neurons in the brain, or even animals in the wild. 
Our model is a generalization of the  {\bf beeping model} \cite{alon2011, CK10}, where agents can only emit a single frequency. 

One of the most studied problems in distributed computing is to select a
{\bf maximal independent set (MIS)} in the communication network \cite{luby1986, ALON1986567, COLLIER1996429, KMW04, KMNW05, Metivier11, Wan04, MW05}.  An MIS is a subset of agents, such that no two selected agents are adjacent, and such that the inclusion of any additional agent in the set will cause a conflict, where two selected agents are adjacent. 
In the scenario where agents are animals in the wild, selecting an MIS corresponds to the natural phenomenon of choosing local leaders, where the neighborhood structure of the network defines competitive relationships.  The goal is to reach a global state where each agent is either the leader of its neighborhood (in the MIS) or is dominated by a leader in its neighborhood.
{\bf The notion of competition is  at the heart of our protocols where the number of notes each agent sings reflects the strength of the agent and its likelihood to dominate its neighbors.} 
Much of the interest in selecting an MIS in a distributed context comes from
its many applications in networking, especially radio sensor networks,
where an MIS is as a building block for routing and clustering \cite{Peleg00}. 
In the biological setting,  a variant of distributed MIS selection is performed by flies \cite{Beep11}.

There are two  communication variants in the singing model which we consider: 
\begin{itemize}
    \item {\bf Singing communication:} an agent can hear the union of notes sung by its neighbors.
    \item {\bf Self-jamming singing communication:} an agent can hear the union of notes sung by its neighbors except the notes sung by the agent itself.
\end{itemize}
This distinction is also considered in the beeping model, where the default version  is self-jamming in the sense that an agent must decide whether to beep or be silent in a given round. In this case, an agent can only hear a beep or the absence of a beep if it chooses to be silent during the round. A version corresponding to our non-self-jamming communication model is also considered in \cite{alon2011} 
(which they call {\em sender-side collisions}), where an agent can beep and listen for a beep at the same time.

\subsection{Summary of our results}

Agents  operate in rounds of simple computation and communication. 
During a round, each agent sings a set of notes  and listens for a set of notes
and then updates its state at the end of a round. 
In the simplest version of our model, 
the rounds are {\bf synchronous}, meaning that the rounds for each agent begin and
end at the same time. Although there may be some delay for a transmission from an
agent to reach a neighboring agent, we assume in this synchronous version of the model that all notes sung by an agent in a given round will be heard by its neighbors by the end of the round.
We also start by assuming that the network is {\bf static} and each agent is initialized to the same internal start state.

We introduce and analyze a {\bf singing protocol} and a {\bf self-jamming singing protocol}, shown in Figures \ref{fig:synchAlg} and \ref{fig:synchAlgSJ}, respectively.
The singing protocol is easier to describe and analyze, while the self-jamming singing protocol is more general, as it also applies to the more challenging self-jamming variant of the singing model.
Both are oblivious protocols that are very simple and natural.  
The only memory stored by an agent from one round to the next is the state of the agent, which is one of three possibilities: 
$\{\outset, \inset,  \unset\}$. An agent can also generate random bits.
We show that our protocols converge to
an MIS in time $O(\log(n))$, where $n$ is
the number of agents in the network.


Furthermore, we prove that these two protocols without any modifications converge to an MIS under the following two generalizations of the model:

\begin{itemize}
    \item{\bf Asynchronous rounds:} Agents operate in asynchronous rounds, meaning that the start times and durations of their rounds may differ from those of other agents. In addition, there may be delays in the transmission of sung notes from a singer to a listener. These transmission delays and round durations can also vary over time for a single agent. We assume that the minimum round duration for any agent is at least twice the maximum transmission delay between any pair of agents.  Under this assumption, we prove that both the singing protocol and the self-jamming singing protocol will converge to an MIS in time 
    $O(\log (n) \cdot T)$, where a unit of time is the maximum round duration over all agents, and $T$ is the ratio between the maximum length of a round and the minimum length of a round. 
    \item {\bf Dynamic networks:} We generalize the static network model to a dynamic, fault-prone setting where both agents and edges can be added or removed over time. New agents may join in arbitrary states, and faulty agents (those that arbitrarily change state) can be modeled as agents being deleted and reinserted with a different state. Both protocols (singing and self-jamming) remain fully fault-tolerant in this dynamic model in a strong sense: any change in the network affects only a bounded local neighborhood. We classify edges as either {\bf active} or {\bf eliminated}, 
    we prove that changes in the network can only cause
    eliminated edges that are within distance three of a change to become active, and we show that the number of active edges shrinks by a constant factor in a constant number of rounds. 
    Therefore, if there are initially $a$ active edges and no further changes occur, the same singing protocol designed for the static, synchronous model will still converge to a Maximal Independent Set (MIS) in time $O(\log(a))$.
\end{itemize}



\section{Related Work}

In all of the models we review here,
the agents of the network operate in rounds of simple computation and communication. Agents broadcast signals in each round that are heard by their neighbors.
Models vary according to the computational capabilities of the agents,   the complexity of the messages they can send, the extent to which they have knowledge of the communication network, whether rounds are synchronous or asynchronous, whether the communication network is static or dynamic, and whether agents can exhibit faulty behavior. 

The most commonly studied family of models in this scenario is message passing \cite{L92, Peleg00, S13}.
In the synchronous version of the message passing model, in each round an agent can send a different message to each of its neighbors in the network, interpret the messages received from each neighbor, and perform an arbitrary computation. 
These models assume that each agent knows a unique identifier for each of its neighbors.
One of the most influential algorithms for MIS in the standard message passing model is Luby's algorithm \cite{luby1986}, which is discussed in more detail below as it provides a basis for the protocols we introduce. 

One of the first distributed models in which agents do not need information about their neighborhood is the Radio network model \cite{CK85}. This model assumes a very strict collision policy in which an agent can hear a message from its neighbors only if exactly one of its neighbors is broadcasting.
Moscibroda and Wattenhofer present an protocol for maximal independent set where the communication network is a unit disk network \cite{MW05}. Schneider and Watenhofer give an protocol that converges in $O(\log^* n)$ rounds for the more general polynomially bounded growth networks \cite{SW08}.

{\bf Our work is inspired by the beeping model} \cite{CK10, alon2011}, in which agents can choose either to listen or to emit a beep. When an agent listens, it can only distinguish between silence (no neighbors are beeping) and noise (at least one neighbor is beeping). In our terminology, this default setting corresponds to the self-jamming model. Some variants of the beeping model allow sender-side collision detection, which aligns with our non-self-jamming model, where an agent can both listen and beep simultaneously.

The beeping model was introduced in \cite{CK10}, which studied a variant of network coloring for collision-free message transmission. Subsequent work presented maximal independent set protocols \cite{alon2011, Beep11} that achieve polylogarithmic convergence times under various assumptions, such as known upper bounds on network size or uniform initial configurations.

Although agents in the beeping model receive finite information per round, they are not finite automata. They use counters that grow over time, making their memory usage dependent on the convergence time, which itself depends on the size of the network. Similarly, agents in the singing model are not finite automata, 
since some agent sends $\Omega(\log(n))$ notes in a round with high probability. Unlike the beeping model, only a finite amount of information at an agent is carried between rounds in our singing protocols.

The beeping model assumes synchrony, meaning that all agents share the same start time and round duration. Protocols in this model can be adapted to settings where round durations are fixed but start times vary, at the cost of doubling the convergence time. However, {\bf the beeping protocols have not been extended to handle fully asynchronous rounds}. While agents in the beeping model may wake up and begin participating at different times, the network itself remains static. Moreover, each agent is assumed to start in a specific initial state.  Thus, {\bf the beeping protocols have not been extended to handle dynamic network settings where an adversary may alter internal agent states.} 

A more recent model called the {\it Stone Age} (SA) model was introduced in \cite{EW13} by Emek and Wattenhofer, where
each agent in a communication network is a finite automaton. 
In some sense, the SA model can be seen as a multi-frequency, asynchronous generalization of the beeping model, although, as discussed above, the beeping model protocols are not necessarily finite computations. 
One drawback of the SA model is that the communication is managed by each agent having a
receiving port for each of  its neighbors. As a result,
the storage requirements per agent scale linearly with the degree of the agent in the Stone Age model,
and thus {\bf the Stone Age model is not network obvlivious}.

Since in the SA model, each agent has finite computational abilities,
the set $\Sigma$ of possible messages that can be sent from an agent is finite.
An agent periodically activates the execution of a protocol in which it uses information received from neighboring agents as well as an internal state to determine its new state and a new message to broadcast to its neighbors.
The activation times of the agents are  asynchronous and messages can incur a delay in reaching a neighbor. The intervals between activation times of an agent as well as delay times can be adversarially chosen. The convergence of a protocol is measured in units of time, where a time unit upper bounds the length of any delay time or the time between consecutive activations of the same agent. 

The SA communication model is that each
 agent $v$ has a port $\psi_v(u)$ for each of its neighbors $u$ to store messages received. When an agent $u$ sends a message $\sigma \in \Sigma$ to agent $v$, and $\sigma$ arrives at $v$, it overwrites whatever happens to be in port $\psi_v(u)$.
Since the degree of the network may be large and the computational capabilities of the agents are finite, an agent cannot distinguish between all of its ports.
At an activation time for agent $v$, the only information $v$ can use is whether any of its ports contain $\sigma$, for each $\sigma \in \Sigma$. In some versions of the model, agent $v$ can learn the number of its ports that contain $\sigma$ up to some finite cap $b$.
Note that since each message sent from a particular neighbor overwrites the previous message sent by that neighbor, it is not clear how the communication model can be implemented without hardware at each agent that scales with the degree. 
The original paper introducing the Stone Age model \cite{EW13} gives a protocol to find a maximal independent set that converges in time $O(\log^2 n)$.
They also solve other problems such as tree coloring and maximal matching in poly-logarithmic time. 

Emek and Uitto \cite{EU20} extend the Stone Age model to allow for a dynamically changing network. 
They prove that changes in the network only create local disruptions to the protocol,
in a similar manner as we do here.
They provide an MIS protocol such that after the protocol stablizes, an agent which does not experience a nearby change in the network does not change its state. Meanwhile the time for an agent which is affected by a change to stabilize is bounded by $O(\log^2 n)$.

Critically, the SA model does not allow 
agents to exhibit faulty behavior. In particular, a newly added agent, must start out in a well-defined initial state.
Emek and Keren \cite{EK21} extend the SA model to a more general model in which agents can exhibit faulty behavior. They call protocols which converge in this environment {\it self-stabilizing} and give protocols which converge in time $O(\log n \cdot (D + \log n))$,
where $D$ is the diameter of the communication network.
Thus, {\bf the Stone Age protocols do not converge in logarithmic time in a dynamic network.}
In contrast, our protocols converge in $O(\log n)$ time 
in a dynamic network. 

The MIS result in \cite{EW13} makes use of a 
reduction showing that it is valid to assume  that the activity of the agents is roughly synchronized into rounds so that the information used by agent $v$ at the end of its $t^{th}$ round depends only on the information sent by its neighbors during their round $t-1$.
The reduction from the asynchronous to synchronous model requires a new protocol whose
state space and message set increases by a constant factor.
A general reduction from synchronous to asynchronous environments of this kind may be possible in the singing model, but we prefer to prove our results separately in the asynchronous environment since our simple protocol works as-is, without any modification.





\section{Protocols and Analyses Overview}

The MIS algorithm in \cite{luby1986} is our starting point: in each round, each agent chooses a random number uniformly from $[0,1]$ and adds itself to the MIS if its number is the largest among its neighbors, thus eliminating the agent and all its neighbors from all subsequent rounds.  An analysis shows that a constant fraction of edges are eliminated in each round, and thus the algorithm completes in $O(\log(n))$ rounds.

In our protocols, each agent chooses a number bit-by-bit, from the high-order bit to the lower-order bits, continuing as long as each chosen bit is a $1$, terminating at the first choice of a $0$ bit.
The number chosen by an agent is represented by two bits on average, and by at most $O(\log(n))$ bits in the worst case with high probability.  The chosen number is a rough approximation of a uniformly chosen random number, which can lead to ties among neighboring agents.  Our analysis shows that ties do not materially affect the analysis, and a constant fraction of the edges are eliminated in each round.

A significant difference between our work and that in~\cite{luby1986} is that it can take several rounds for information to propagate between agents.
For example, if an agent chooses the largest number among its neighbors in a round and adds itself to the MIS,  neighboring agent $z$ only hears about this event in the next round and only stops participating in the round after that.
Thus, agent $z$ is a {\bf \em zombie} since it participates for several rounds after its fate has already been determined, thus affecting decisions made by its neighbors during those rounds.
We track the propagation of information by analyzing our protocols in steps consisting of several rounds.





\subsection{Singing Protocol}

Each agent is in one of three states
$\{ \outset, \inset,  \unset\}$, depending on whether the agent is 
out of or in the independent set, or undecided.
In the case of a static  network, we assume that every agent starts in state $\unset$.
The singing protocol executed by every agent during each round is shown in Figure \ref{fig:synchAlg}.

\begin{figure}[ht!]
\begin{center}
\noindent\fbox{\begin{minipage}{\textwidth}
\begin{tabbing}
\= ~{\sc Round protocol for an agent in state $s$}\\
\> ~~~~~~~\= Compute new $\ell$-value by flips of a random coin\\
\= ~{\sc Sing during round}\\
\>~~~~~~~\= If $s = \outset$ then do not sing\\
\>~~~~~~~\= If $s = \inset$ then sing note $0$\\
\>~~~~~~~\= If $s = \unset$ then sing notes $1,\ldots,\ell$\\
\=~ {\sc Listen during round}\\
\>~~~~~~~\= Listen for note $0$\\
\>~~~~~~~\= If $s = \unset$ then also listen for note $\ell$\\
\=~ {\sc At end of round -- set state for next round}\\
\> ~~~~~~~\= If $s = \inset$ then\\
\>\> ~~~~~~~\= If heard note $0$ then set $s = \unset$\\
\>\> ~~~~~~~\= Else leave $s = \inset$\\
\> ~~~~~~~\= If $s = \unset$ then\\
\>\> ~~~~~~~\= If heard note $0$ then set $s = \outset$\\
\>\> ~~~~~~~\= Else\\
\>\>\> ~~~~~~~\= If heard note $\ell$ then leave $s = \unset$\\
\>\>\> ~~~~~~~\= Else set $s = \inset$\\
\> ~~~~~~~\= If $s = \outset$ then\\
\>\> ~~~~~~~\= If heard note $0$ then leave $s = \outset$\\
\>\> ~~~~~~~\= Else set $s = \unset$\\
\end{tabbing}
\end{minipage}}
\end{center}
\caption{Singing protocol}
\label{fig:synchAlg}
\end{figure}

At the beginning of a round, 
each agent selects its $\ell$-value by flipping a 
random unbiased $\{0,1\}$-valued coin until there is a $0$ outcome,
and then setting $\ell$ to the number of such flips. 
For example, if the first coin flip outcome is $0$ then $\ell=1$.
If instead the sequence of coin flip outcomes is $1, 1, 1, 1, 0$ 
then $\ell = 5$.

Agents sing note $0$ if and only if the state of the agent is $\inset$.
All agents listen for note $0$ to detect whether they are adjacent to an $\inset$ agent.
Meanwhile $\outset$ agents remain silent, and
$\unset$ agents sing notes $1$ through $\ell$ 
and listen for note $\ell$.
At the end of a round, each agent updates its state based on its current state, the notes it heard, and possibly its $\ell$-value.
The transition logic is shown in Figure \ref{fig:singTrans}.

\begin{figure}[ht]
\begin{center}
\includegraphics[width=.6\linewidth]{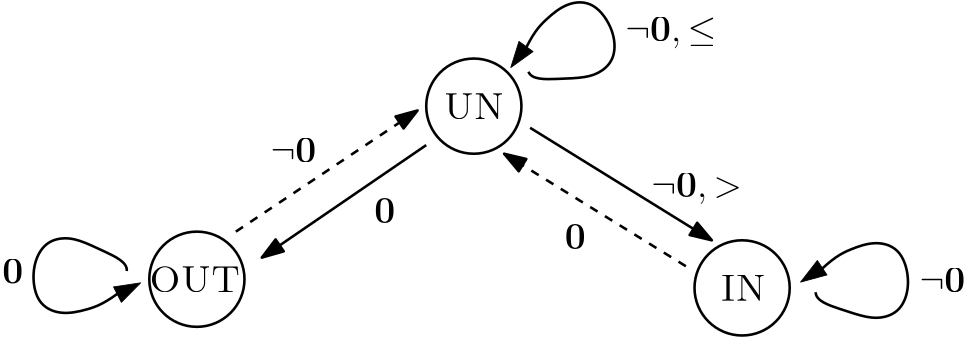}
\caption{Transition logic for the singing protocol. The label $0$ on an arrow indicates that the agent hears note $0$, and $\neg 0$ indicates that the agent does not hear note $0$ in the round. 
Recall that note $0$ is heard if and only if the agent has a neighbor in the $\inset$ state.
The label $>$ on an arrow indicates that the $\ell$-value for the agent is strictly larger than that of all of its $\unset$ neighbors, and $\le$ indicates that the $\ell$-value for the agent is less than or equal to the $\ell$-value of at least one of its  $\unset$ neighbors.
The dotted lines show transitions that will not arise in the case of a static network.}
\label{fig:singTrans}
\end{center}
\end{figure}

If an $\unset$ agent hears note $0$, signifying that it has an $\inset$ neighbor,
it transitions to the $\outset$ state.
If an $\unset$ agent does not hear note $0$ then it competes against its neighbors to enter the $\inset$ state. In particular, if its $\ell$-value is strictly larger than the $\ell$-value of its neighbors (detected by the fact that it does not hear a note corresponding to its own $\ell$-value), then it transitions to $\inset$. Otherwise, it stays $\unset$.

For the sake of analysis, we will further refine the state of an agent in a round with superscripts $\inset$ and $\overline{\inset}$.
A superscript of $\inset$ indicates the agent has a neighbor in state $\inset$ in the same round. A superscript of $\overline{\inset}$ indicates the agent does not have a neighbor in state $\inset$ in the same round.
At the end of a round, an agent knows whether it has a neighbor that is $\inset$ since all agents listen for note $0$ in every round.

Note that in a static network, there will never be a conflict with an $\inset$ agent that is adjacent to another $\inset$ agent (signified by an $\inNin$ agent). There will also never be an $\outNnin$ agent, representing an $\outset$ agent that does not have any $\inset$ neighbors.
The logic for these states are included in the protocol since these situations can occur in dynamic networks. 

The following definition is used to differentiate between agents and edges whose final state is still undetermined from those whose final state is known. 
\begin{definition}
\label{def:active}
    An agent is {\bf eliminated} if it is $\inNnin$ or is adjacent to an agent that is $\inNnin$.
    Otherwise the agent is said to be {\bf active}.
    An edge is {\bf eliminated} if either endpoint is eliminated, otherwise the edge is {\bf active}.
    $\calV_i$ is the set of active agents in round $i$, and $\calE_i$ is the set of active edges in round $i$.
\end{definition}
All of the bounds in this paper show that in a finite number of rounds, the expected number of active edges is reduced by a constant factor.

\subsection{Analysis Overview: Static, Synchronous, Non-Self Jamming}

Consider an active edge $\{x,y\}$ at the beginning of  round $i+1$ pictured below. 
$\unset$ agents are shown with solid lines and $\outset$ agents are shown with dotted lines.
In a static network, every active agent is in the state $\unNnin$. Since the edge $\{x,y\}$ is active, both $x$ and $y$ are $\unNnin$.
\begin{center}
\includegraphics[width=.35\linewidth]{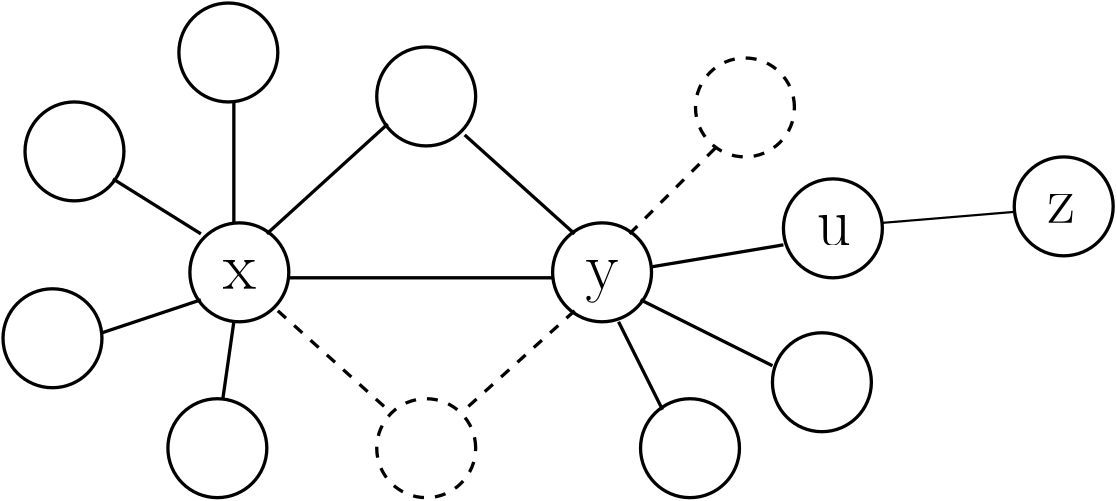}
\end{center}

If agent $x$ wins its competition in round $i+1$, i.e., its $\ell$-value is strictly larger than it's six $\unset$ neighbors, then $x$ will transition to $\inset$, while $y$ will stay $\unset$.
In this case both $x$ and $y$ are eliminated at the end of round $i+1$, as are all the edges incident to $y$. 
(The edges incident to $x$ are also eliminated, but we do not give $x$ credit for that here in order to avoid overcounting.)
It's possible that $y$ has other neighbors that also transition to $\inset$ in round $i+1$, so we give $x$ sole credit for eliminating $y$'s incident edges if $x$ has a larger $\ell$-value than its own neighbors as well as larger than all of $y$'s neighbors as well.

Let $N(x)$ be the set of $x$'s neighbors in state $\unset$  and $N(y)$ be the set of $y$'s neighbors in state $\unset$ during round $i+1$.
Let $|N(x)| = d_x$ and $|N(y)| = d_y$.
In the diagram above, $d_x = 6$ and $d_y = 4$.
Agent $x$ will get credit for eliminating the edges incident to $y$ if its $\ell$-value is strictly larger than the $\ell$-value of every agent in the set $N(x) \cup N(y) - \{x\}$.
In Section \ref{sec:coin} we analyze the coin flipping process and show that the probability of this event is at least
$$\frac 2 3 \cdot \frac{1}{|N(x) \cup N(y)|} \ge \frac 2 3 \cdot \frac{1}{d_x + d_y}.$$
We would like to give $x$ credit for eliminating $d_y/2$ edges in this case. The factor of $1/2$ comes from the fact  that an edge incident to $y$ could also be eliminated because the other (non-$y$) endpoint  became eliminated. For example, if agent $z$ also transitions to $\inset$ in round $i+1$, then the edge $\{y,u\}$ is eliminated from each endpoint, and we are over-counting by a factor of $2$.

Thus, the possibility of $x$ transitioning to $\inset$ during round $i+1$
causes an expected decrease to the number of active edges incident to $y$ to be at least
$$\frac{d_y}{2} \cdot \frac 2 3 \cdot \frac{1}{d_x + d_y} = \frac{d_y}{3(d_x + d_y)}.$$
Similarly, the possibility of $y$ transitioning to $\inset$ during round $i+1$ causes an expected increase to the number of active edges  incident to $x$ to be  at least
$d_x/3(d_x + d_y)$. The total expected decrease to the number of active edges due to the edge $\{x,y\}$ is $1/3$. Summing up over all the active edges would yield an expected decrease of $|\calE_{i+1}|/3$, where $\calE_{i+1}$ is the set of active edges in round $i+1$.

However, there is a problem with this argument illustrated in the diagram below. Agent $v$ is shown in the diagram with a double circle to indicate that it is in the $\inset$ state during round $i+1$.
\begin{center}
\includegraphics[width=.35\linewidth]{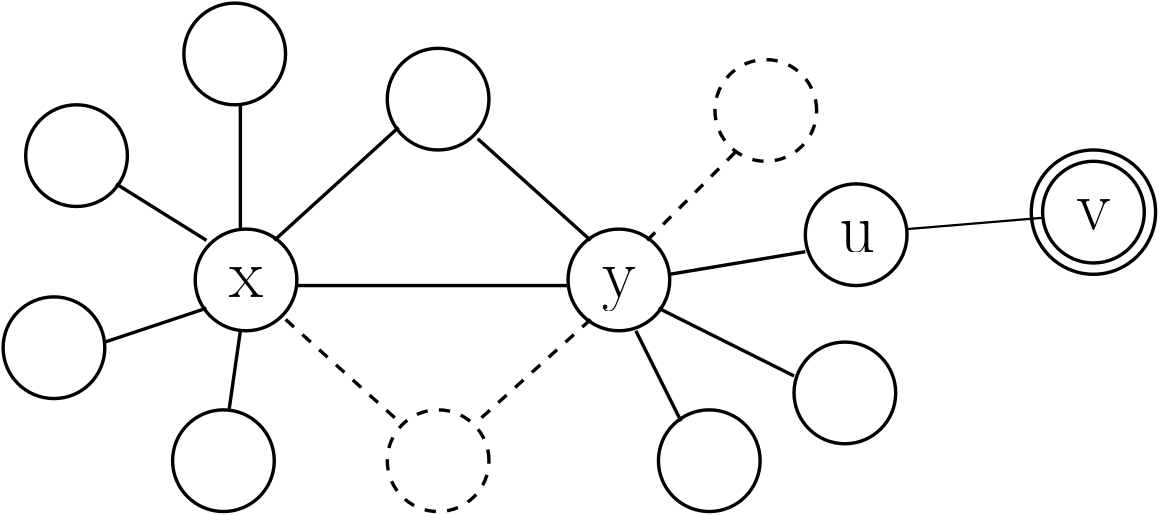}
\end{center}
Since $u$ is in the $\unset$ state in round $i+1$, agent $x$ still must compete with $u$ to get credit for eliminating edges incident to $y$, which decreases the probability that $x$ gets credit for eliminating the edges incident to $y$.
However, since $v$ is in the $\inset$ state, the edge $\{y, u\}$ is already eliminated before round $i+1$. Therefore $x$ should not get credit for eliminating edge $\{y, u\}$
in round $i+1$.
Note that in this case $u$ is a zombie agent, whose fate has been determined, although agent $u$ is unaware of this fact during round $i+1$.

The key insight to handle this issue is that being a zombie agent is a transient condition, due to the delay in propagating information in the network. 
By the end of round $i+1$, $u$ will have heard note $0$ sent by $v$ and thus be in the $\outset$ state in round $i+2$. 
This same reasoning also means that agent $v$ must have been $\unset$ during the previous round $i$, which in turn means that the edge $\{y,u\}$ was active during round $i$
and only became eliminated in round $i+1$.

We analyze the process across two rounds, formalized in Lemma~\ref{lemma:snapshot}. Let $e_i = |\calE_i|$ be the number of active edges in round $i$.  
We can count the number of edges that have active or zombie agent endpoints in round $i+1$ that are
eliminated by edges in $\calE_{i+1}$, 
which by the reasoning
above is at least $e_{i+1}/3$.
Note that all these eliminated edges were active in round $i$.
Therefore,
the expected number of edges that were active in round $i$ and are eliminated by round $i+2$ is at least $e_{i+1}/3$:
$$E[e_{i+2}] \le e_i - \frac{e_{i+1}}{3}$$
In addition, the number of active edges never increases: $e_{i+2} \le e_{i+1}$. The minimum of the two bounds is maximized when $e_{i+1} = \frac{3}{4} \cdot e_i$, and the result is that
$E[e_{i+2}] \le \frac{3}{4} \cdot e_i$.

Some form of the above argument is used in all of the proofs when generalizing to  dynamic networks, asynchronous rounds, and the self-jamming protocols. 

\subsection{Generalizing to Dynamic Networks}

The only agents that are active in a static network are those in the state $\unNnin$.  
There are two other  possibilities that can only be reached if there is a change in the network. The first is an $\outset$ agent that is not adjacent to an $\inset$ agent, making the agent $\outNnin$.
The other possibility is that $v$ is adjacent to a conflicted agent that is $\inNin$.
We call these last two possibilities {\em unnatural} states because they will not be reached in a static network.

We show that if an agent is not near a change in the graph, then within two rounds it will no longer be in an unnatural state. 
In particular, conflicted $\inNin$ agents will drop back to the $\unset$ state after one round.
Then any $\outset$ agents that are not adjacent to an $\inset$ agent, we transition back to $\unset$ in the next round. 
In the third round, an active edge that is unaffected by a change in the graph will have endpoints that are both in the $\unset$ state with no $\inset$ neighbors and will be able to compete as they did in the static case.

We analyze the protocol over a sequence of three consecutive rounds.
Consider an active edge   
in the  third round. By the reasoning above, if
$\{x,y\}$ is  not near a change in the graph, then both $x$ and $y$ will be in the $\unset$ state with no $\inset$ neighbors. They are ready to compete in the third round to transition to $\inset$.
Lemma \ref{lem:settledown}  establish that if $y$ has a neighbor $u$ in the $\unset$ state during the third round (meaning that $x$ or $y$ is competing against $u$), then the edge $\{u,y\}$ was also active at the beginning of the sequence of three rounds. This allows us to credit $x$ or $y$ with the elimination of edge $\{u,y\}$.

If $A_i$ is defind to be the set of edges that are affected by a change at the beginning of round $i$, then 
Theorem \ref{th:dynamic} says that:
$$E[\calE_{i+1}] \le \frac 3 4 \cdot \calE_{i} +  |A_{i+1} \cup A_{i+2} \cup A_{i+3}|.$$

\subsection{The Self-Jamming Protocol}

In the self-jamming model, an agent cannot listen for any of the notes that it is singing in a given round.
The protocol for the self-jamming model is given in Figure \ref{fig:synchAlgSJ}.

Self-jamming is managed by having $\unset$ agents sing only odd-numbered notes and $\inset$ agents sing only even numbered agents.
In each round, each agent that is not $\outset$ {\em competes against} its neighbors that have the same state, so the $\unset$ agents compete against their neighbors that are also $\unset$, and an $\inset$ agent
competes their neighbors who are also $\inset$.
An agent sings note $0$ if and only if the agent is $\inset$.
Since $\unset$ agents only sing odd-numbered notes, the $\unset$ agents will always hear a $0$ notes sung by an $\inset$ neighbor.
By contrast, $\inset$ agents cannot hear a $0$ note sung by an $\inset$ neighbor. 
The transition logic for the protocol is given in Figure \ref{fig:singTransSJ}.

\begin{figure}[ht]
\begin{center}
\includegraphics[width=.6\linewidth]{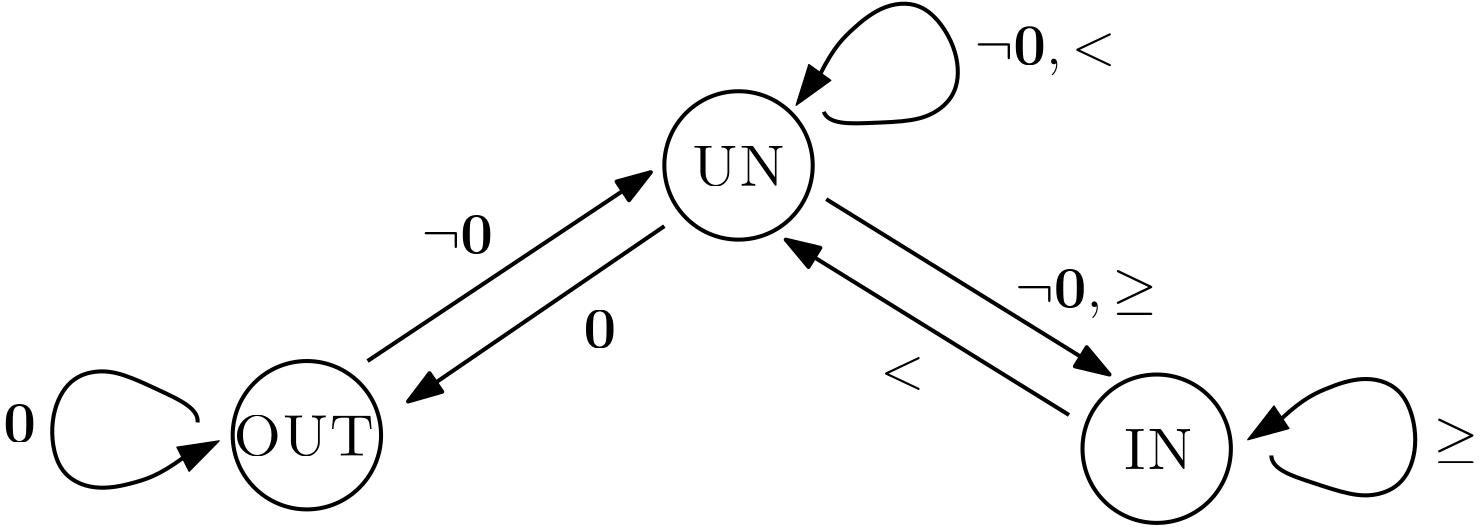}
\caption{Transition logic for the singing self-jamming protocol. The label $0$ on an arrow indicates that the agent hears note $0$, and $\neg 0$ indicates that the agent does not hear note $0$ in the round. 
Recall an agent sings note $0$ if and only if the agent is in the $\inset$ state.
For $\inset$ agents, an outgoing arrow labeled $<$ indicates that its $\ell$-value is strictly less than the $\ell$-value of one of its $\inset$ neighbors.
An outgoing arrow labeled $\ge$ indicates that its $\ell$-value is  greater than or equal to  the $\ell$-values of all of its $\inset$ neighbors.
For $\unset$ agents, an outgoing arrow labeled $<$ indicates that its $\ell$-value is strictly less than the $\ell$-value of one of its $\unset$ neighbors.
An outgoing arrow labeled $\ge$ indicates that its $\ell$-value is  greater than or equal to  the $\ell$-values of all of its $\unset$ neighbors.
}
\label{fig:singTransSJ}
\end{center}
\end{figure}

The $\unset$ agents compete with each other to determine whether they should be in the independent set.
Any $\unset$ agent that is adjacent to an $\inset$ agent will hear the note $0$ and will transition to $\outset$.
Otherwise, if two adjacent $\unset$ agents have different $\ell$-values, then the one with the smaller $\ell$-value will hear the note  from the agent with the larger $\ell$-value and will not transition to $\inset$.
If an $\unset$ agent does not hear from an $\unset$ agent with a higher $\ell$-value, it transitions to $\inset$. A conflict occurs if two adjacent agents have the same $\ell$-value, and neither of them has a neighbor with a strictly larger $\ell$-value. In this case both will transition to $\inset$.

Because of the possibility of conflicts, even when the network is static, the $\inset$ agents must continually listen for other $\inset$ agents to determine if they should drop out of the independent set. 
The $\inset$ agents all transmit even-numbered notes.
Similar to the case of the $\unset$ agents, if two adjacent $\inset$ agents have different $\ell$-value, then the one with the smaller $\ell$-value will hear the note  from the agent with the larger $\ell$-value and will drop out of the independent set by transitioning to $\unset$, resolving the conflict between the two agents.
If an $\inset$ agent does not hear from an $\inset$ agent with a higher $\ell$-value, it remains $\inset$.
Conflicts can persist through time if two adjacent $\inset$ agents continue to have the same $\ell$-value. Since each agent randomly selects a new $\ell$-value in each round, we can bound the probability that a conflict persists through time.
The full analysis for the Self-Jamming protocol appears in Section \ref{sec:SJ}. A key part of the argument is Lemma \ref{lem:elig} which 
says that if an agent $v$ is active in round $i$ and $p$ is the probability that two rounds later  $v$ has at least one neighbor that is in the $\inset$ state, then the probability that $v$ has a neighbor that is $\inNnin$ (making $v$ eliminated) is at least $p/2$.
This bounds the probability that $v$'s behavior is affected by a conflicted $\inNin$ agent.

\subsection{Generalizing to Asynchronous Rounds}

We analyze the protocol for the asynchronous case over an interval that starts at time $t$
and ends at time $t' = t+ 3 \cdot T_{max} + 2 \cdot \Delta_{max}$, where $T_{max}$  is the maximum length of a round and $\Delta_{max}$ is the maximum delay  from  when a note is sung and when it is heard by a listening neighbor.
Time is continuous, and now $\calE_t$ denotes the set of active edges at time $t$.
As before, we lower bound the number of edges that transition from active to eliminated in the course of the interval. 
We define an intermediate point  $\ta = t + 2 \cdot T_{max} + \Delta_{max}$ in the interval and define for each active agent $x$, its interval $(x, i_x)$ that spans time $\ta$. 
This time is chosen so that any round that began before time $t$ will have finished and their notes will have reached their destinations before $(x, i_x)$ begins. The separation allows us to randomize only over coin flips generature during rounds that start after time $t$.

In the synchronous case, we gave $x$ credit for eliminating the active edges incident to $y$ if its $\ell$-value dominates the $\ell$-values of all the agents in $N(x) \cup N(y) - \{x\}$.
In the asynchronous case, we need to consider all of the rounds for all of those agents which begin in the interval from $t$ to $t'$. Thus, the number of rounds $x$ is competing against is 
$|N(x) \cup N(y)| $ times a factor of $O(T_{max}/T_{min})$. Thus,  the extra factor of
$O(T_{max}/T_{min})$  in the convergence time.

In the synchronous case, the event that an agent $x$ is in the $\unNnin$ state and ready to compete in round $i$
is independent from the all random choices made during round $i$. The random choices made during round $i$ by $x$ and its neighbors determine whether $x$ wins. However, in the asynchronous case, these events are potentially correlated in complex ways due to the fact that a note sung by agent $v$ during a round could be heard by agent $x$ in round $(x, i_x)$ as well as before round $(x, i_x)$.
In particular, the rounds that span the start of round $(x, i_x)$ will affect 
whether $x$ is ready to compete in that round as well as how well $x$ fares in the competition.
Most of the complexity in the proof results from carefully separating these dependencies.

\subsection{Discussion}

We have proven that the Singing protocol and the Self-Jamming Singing protocol converge in logarithmic time for dynamic graphs and separately for asynchronous rounds. We have not proven convergence in the setting that combines the two: asynchronous rounds with a dynamic communication network. It is likely that  the techniques for the two generalizations can be combined, but we leave that for future work.

In our protocols, agents compute and sing a varying
number of notes and listen for at most two particular
notes in a round.  In some scenarios it could be more expensive
for agents to sing notes than it is to listen for notes. 
There is a simple transformation to the protocols 
we describe here where
an agent sings at most two notes in each round but then needs to 
listen for all possible notes in the
round. An issue with this approach is that agents do
not know an apriori bound on the number of notes to listen for.
We leave it as an open problem to determine whether it is possible to devise protocols in the singing model 
where agents sing a constant number of notes
while only having to listen for a bounded number of notes as well.


\section{Analysis for Synchronous Rounds}

For each agent $v$ in the network, for $i \ge 0$, we let $(v,i)$ 
denote round $i$ for agent $v$.  
Let $s_{v,i}$ denote the state of agent $v$ during $(v,i)$.
We assign labels to edges corresponding to their endpoints,
so an edge label is $(\unset,\unset)$ in a round if both its endpoints
are in state $\unset$ in the round, or an edge label is 
$(\unNnin, \unNnin)$ in a round if the status of both its endpoints
are $\unNnin$ in the round, or an edge label is $(\unNnin, \unNin)$
if one endpoint is $\unNnin$ and the other 
is $\unNin$ in the round. In this notation the ordering of endpoints doesn't matter.
The number of active edges during round $i$ ($|\calE_i|$) will be denoted by $e_i$.

\subsection{Static Communication Networks}

For the analysis in this section we assume that the network is static and
the initial state of all agents is $\unset$. 
We start with a few lemmas which restrict the set of states that can be reached when the communication network is static.

\begin{claim}
    If the communication network is static, once an agent reaches the $\inset$ or $\outset$ state, its state never changes again. Moreover, there are no $\inNin$ agents.
\end{claim}

\begin{proof}
    By induction on the number of rounds.
\end{proof}

\begin{claim}
\label{cl:contain}
    In a static network, any active agent must be $\unNnin$. Furthermore, any $\unset$ agent in round $i$ must have been active in round $i-1$.
\end{claim}

\begin{proof}
The definition of an active agent is that it is not $\inNnin$ and has no neighbors that are $\inNnin$.
    Since all $\inset$ agents are $\inNnin$, an agent is active if and only if it is not in the $\inset$ state and does not have any neighbors that are $\inset$. 
    An active agent cannot be in the $\outset$ state since the only way to transition to the $\outset$ state is to have an $\inset$ neighbor. Therefore all active agents are in the $\unset$ state and have no $\inset$ neighbors.

    For the second  part of the claim, observe that any agent that is $\unset$ in round $i$ was also $\unset$ in round $i-1$. Any agent that was $\unset$ in round $i-1$ would have transitioned to the $\outset$ state if it had an $\inset$ neighbor in round $i-1$. Therefore, if an agent is $\unset$ in round $i$, it must have been $\unset$ in round $i-1$ with no $\inset$ neighbors, making it an active agent in round $i-1$.
\end{proof}

\begin{lemma}  
For any round $i$,
\begin{equation*}
\Exp[ e_{i+2} ] \le \frac{3}{4} \cdot e_i.
\end{equation*}
\label{lemma:snapshot}
\end{lemma}

\begin{proof}
Let $\Vunset_{i+1}$ be the set of $\unset$ agents in round $i+1$ 
and let $\calU_{i+1}$ be the set of edges in the network
induced by the agents $\Vunset_{i+1}$.
Then, $$\calE_{i+1} \subseteq \calU_{i+1} \subseteq \calE_i.$$
The first containment follows from the first part of Claim \ref{cl:contain}: if two endpoints of an edge are active in round $i+1$, then those two endpoints are also in the $\unset$ state.
The second containment follows from the second part of Claim \ref{cl:contain}:
if the two endpoints of an edge are both $\unset$ in round $i+1$, then the two endpoints were both active in round $i$.

Consider an edge $\{x,y\} \in \calE_{i+1}$. 
If $x$ transitions to $\inset$ in round $i+2$ 
then $y$ will no longer be $\unNnin$ in round $i+2$. By Claim \ref{cl:contain}, 
this implies that $y$ will no longer be active in round $i+1$, and thus
all the edges in $\calU_{i+1}$ incident to $y$ will no longer be in $\calE_{i+2}$.
The same applies to edges incident to $x$ if $y$ transitions to $\inset$ in round $i+2$. 

 
Let $N(x)$ and $N(y)$ be the set of neighbors of $x$ and $y$, respectively, 
with respect to the network $(\Vunset_{i+1}, \calU_{i+1})$, and 
let $d_x = |N(x)|$ and $d_y = |N(y)|$.  
We first analyze
the expected credit due to $x$ for edge $\{x,y\}$.
We define the event $D[x \rightarrow y]$ in round $i+1$ to be the event that
$\ell_{x,i}$ is greater than the $\ell$-value for any other agent 
in $N(x) \cup N(y) - \{x\}$. 
If $D[x \rightarrow y]$ occurs in round $i+1$, 
then $x$ is $\unNninGT$ in round $i+1$, and thus $x$ transitions to $\inset$ in round $i+2$.
Furthermore, $x$ can take $1/2$ credit for eliminating all the edges in $\calU_{i+1}$ incident to $y$ by round $i+2$,
i.e., no other agent will take credit for eliminating an edge incident to $y$ from its $y$ endpoint. 
(An edge can be eliminated from either of its endpoints.) 
Using  Lemma~\ref{lem:fair}, the probability of 
$D[x \rightarrow y]$ is at least
$\frac{2}{3 \cdot (d_x + d_y)}$.
Thus, the expected credit due to $x$ for edge $\{x,y\}$ is at least
$$\frac{d_y}{2} \cdot \frac{2}{3 \cdot (d_x + d_y)} = 
\frac{d_y}{3 \cdot (d_x + d_y)}.$$

By similar reasoning, the expected credit due to $y$ for edge $\{x,y\}$ is at least
$$\frac{d_x}{3 \cdot (d_x + d_y)},$$
and thus the expected total credit for $(x,y)$ is at least $\frac{1}{3}$.

Adding up the credits over all edges in $\calE_{i+1}$, 
the total expected number of edges from $\calU_{i+1}$ that are eliminated by round $i+2$ is 
at least  $\frac{e_{i+1}}{3}$.

From this, 

\begin{equation}
\label{eq:twornds}
\Exp[ e_{i+2}]  \le |\calU_{i+1}| - \frac{e_{i+1}}{3} \le e_i - \frac{e_{i+1}}{3}.   
\end{equation}

Also, 
\begin{equation}
\label{eq:consec}
 e_{i+2} \le e_{i+1}.   
\end{equation}

From Equations~(\ref{eq:twornds}) and~(\ref{eq:consec}), 
\begin{equation}
\label{eq:min}
\Exp[ e_{i+2}] \le \min\left\{e_{i+1}, e_i - \frac{e_{i+1}}{3}\right\}.  
\end{equation}

The maximum of the minimum on the right in Equation~(\ref{eq:min}) is achieved when $e_{i+1} = \frac{3}{4} \cdot e_i$, from which the proof follows. 
\end{proof}

\subsection{Analysis for Dynamic Networks}
\label{sec:dynamic-synch}

We would like to update the model so that agent and edges can be inserted into or deleted from the network. Even in the model where agents update their status at synchronized time steps, the network can change at any point in time. 
Any edges incident to a new agent are considered to be new edges. Similarly, all the edges incident to a deleted agent are also deleted. A new agent can be initialized to an arbitrary initial state and will begin executing the protocol in the round in which it first appears.
Since the states of the agents are updated synchronously at the end of every round, the effect of a newly added edge only effects the behavior of the agents at the end of the round $i$ where the two neighbors can first hear each others' signal. Therefore, the edge is effectively added at the beginning of round $i$. Similarly, the effect of deleting an edge only effects the behavior of the agents at the end of the round $i$ where the two neighbors first do not hear each others' signal. Therefore, the edge is effectively deleted at the beginning of round $i$. This observation will allow us to assume that edges are added or deleted at the beginning of a round. 
An agent arbitrarily changing state or executing the protocol incorrectly is modeled by deleting and reinserting the agent in the new state, which is effectively a change in the network at the beginning of the next round. Therefore, we will view all changes in the network or faults in the agents as a change of the network occurring at the beginning of a round. 

\begin{definition}[{\bf Changed and Affected Agents and Edges}]
\label{def:affected}
We say that an agent is {\bf changed} in round $i$ if it is incident to an edge that is added or deleted at the beginning of round $i$.
The effect of a changed agent can ripple through the network but only to a bounded distance.
We say that an edge is {\bf affected} by a change in round $i$ in the network if one of its endpoints is  a distance at most $2$ from an agent that is changed in round $i$.
The set of edges affected in round $i$ is denoted by $A_i$.
\end{definition}

\begin{figure}[ht!]
\includegraphics[width=.9\linewidth]{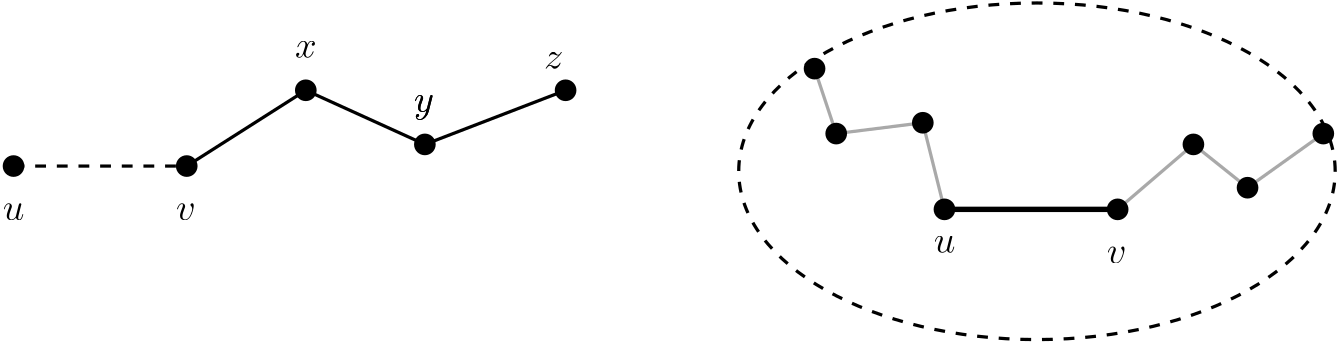}
\caption{In the network on the left, if the edge $\{u,v\}$ is added or deleted at the beginning of round $i$ or if $u$ or $v$ execute their protocol incorrectly at the end of round $i-1$, then the edge $\{u,v\}$ is said to be changed at the beginning of round $i$. In this case, the edges $\{v,x\}$, $\{x,y\}$, and $\{y,z\}$ are affected by the change at the beginning of round $i$.
In the network on the right, if the edge $\{u,v\}$ is not affected by a change at the beginning of round $i$, then all of the agents that are a distance at most $3$ from $u$ or $v$ execute their protocol correctly at the end of round $i-1$ and the network induced by those agents is unchanged at the beginning of round $i$.
}
\label{fig:affectedEdges}
\end{figure}

\begin{claim}
\label{cl:affected}
    (See  Figure \ref{fig:affectedEdges}.) Suppose that an edge $e$ is not affected by a change in round $i$. Let $V_e$ be the set of agents that are a distance of at most $3$ from either endpoint of $e$. Then the network induced by $V_e$ is unchanged at the beginning of round $i$. Furthermore, the agents in $V_e$ executed the protocol correctly at the end of round $i-1$ according to the signals they received.
\end{claim}

We will use the same definition for active and eliminated agents and edges from Definition \ref{def:active} which says that an agent is eliminated if it is $\inNnin$ or has an $\inNnin$ neighbor, otherwise the agent is active. An edge is eliminated if it has at least one eliminated endpoint, and otherwise the edge is active.

We will prove that eliminated edges will remain eliminated unless they are affected by a change in the network. 
We will also show that if $e_i$ is the number of active edges at round $i$, all the edges will be eliminated by time $i + O(\log e_i)$ with high probability, except for those edges  which are effected by a change in the network after round $i$.

The next two lemmas help us characterize the behavior of the algorithm around unaffected edges.

\begin{lemma}
\label{lem:stayElim}
    Suppose agent $v$ is incident to an edge that is not affected in the beginning of round $i+1$. Let $u$ be an agent that is either $v$ itself or a neighbor of $v$.     Then  if $u$ was eliminated in round $i$, $u$ will remain eliminated in round $i+1$.
\end{lemma}

\begin{proof}
If $u$ is eliminated in round $i$, then either $u$ is $\inNnin$ or it has a neighbor $u$ that is $\inNnin$.
    By Claim \ref{cl:affected}, the graph induced by all the agents a distance at most $2$ from $u$ remains unchanged at the beginning of round $i+1$
    and execute the protocol correctly at the end of round $i$.
    Consider a neighbor $x$ of $u$ that is $\inNnin$ in round $i$. Since $x$ does not hear note $0$ in round $i$, $x$ will continue to remain $\inset$. Moreover, all of $x$'s neighbors (which by assumption execute the protocol correctly at the end of round $i$) hear note $0$ from $x$ and will remain not $\inset$. Therefore $x$ will remain $\inNnin$, and $u$ will remain eliminated in round $i+1$.
    \end{proof}

\begin{lemma}
\label{lem:settledown}
    Suppose agent $v$ is incident to an edge that is not affected in the beginning of round $i+1$ or $i+2$. Then if $v$ is active in round $i+2$, then $v$ is in the $\unNnin$ state in round $i+2$.
    Also,if $v$ has a neighbor $u$ that is $\unset$ in round $i+2$, then $u$ was active in round $i$.
\end{lemma}

\begin{proof}
The graph induced by the vertices a distance at most $3$ from $v$ execute the protocol correctly at the end of rounds $i$ and $i+1$.
Therefore, in rounds $i+1$ and $i+2$, this induced graph does not contain any edges with both endpoints in the $\inset$.
Let $V_2$ be the vertices that are a distance at most $2$ from $v$.
No vertices from $V_2$ will ne $\inNin$ in rounds $i+1$ or $i+2$.

    If $v$ is active in round $i+2$, then by Lemma \ref{lem:stayElim}, it must also be active in rounds $i+1$ and $i$.
    Therefore, in rounds $i$, $i+1$, and $i+2$, $v$ is not $\inNnin$ and does not have any $\inNnin$ neighbors. Since there are no $\inNin$ agents,
    in $V_2$ in round $i+1$,
    $v$ must not have any $\inset$ neighbors in round $i+1$ and is itself not $\inset$.
    If $v$ is $\outset$ in round $i+1$, it will transition to $\unset$ in round $i+2$.

    Suppose a neighbor $u$ of $v$ is eliminated in round $i$.
    We will prove that $u$ cannot be $\unset$ in round $i+2$.
    By Lemma \ref{lem:stayElim}, $u$ will remain eliminated in rounds $i+1$ and $i+2$.
    If $u$ itself is $\inNnin$ in round $i$, then $u$ will remain $\inNnin$ in round $i+2$ and cannot be $\unset$.
    If $u$ has a neighbor that is $\inNnin$ in round $i+1$,
    then $u$ will transition to $\outset$ in round $i+2$.
\end{proof}

\begin{theorem}
\label{th:dynamic}
Consider any arbitrary state of the protocol in round $i$. Then
$$\Exp[ |\calE_{i+3}  |] \le \frac 3 4 \cdot  |\calE_i| +  |  A_{i+3} \cup A_{i+2} \cup A_{i+1}| .$$
\end{theorem}

\begin{proof}
Define $\hat{E} = \calE_{i+2} - (A_{i+3} \cup A_{i+2} \cup A_{i+1})$.
We will show that the expected number of edges in $\calE_i$ that flip from active to eliminated is at least $\frac 1 3 |\hat{E}|$. In the worst case, all the edges in
$A_{i+3} \cup A_{i+2} \cup A_{i+1}$ flip from eliminated to active.
Then the net effect on the number of eliminated edges in round $i+3$  will be:
\begin{align*}
     E[|\calE_{i+3}|] & \le |\calE_i| - \frac 1 3 |\hat{E}| + |A_{i+3} \cup A_{i+2} \cup A_{i+1}|
\end{align*}
Meanwhile, by Lemma \ref{lem:stayElim}, the edges outside of $A_{i+3} \cup A_{i+2} \cup A_{i+1}$
that are eliminated in round $i+2$ remain eliminated in round $i+3$. Thus, we have that
$$\calE_{i+3} \le |\hat{E}| + |A_{i+3} \cup A_{i+2} \cup A_{i+1}|$$
The minimum of the two bounds is maximized when $|\hat{E}| = \frac 3 4 |\calE_i|$, which gives a bound of 
$$\calE_{i+3} \le \frac 3 4 |\calE_i| + |A_{i+3} \cup A_{i+2} \cup A_{i+1}|$$
We will show that the expected number of edges in $\calE_i$ that flip from active to eliminated is at least $\frac 1 3 |\hat{E}|$. Using Lemmas \ref{lem:stayElim}, and \ref{lem:settledown}, the analysis will be similar to the static case.

Consider an edge $\{x,y\} \in \hat{E}$. The edge is active in round $i+2$ and is
not affected by a change at the beginning of round $i+1$ through $i+3$.
By Lemma \ref{lem:settledown}, $x$ and $y$ are both $\unNnin$ in round $i+2$.
Define $N(x)$ to be the set of $x$'s neighbors that are $\unset$ in round $i+2$,
and let $|N(x)| = d_x$.
$x$ will compete against all of the agents in $N(x) \cup N(y) - \{x\}$.
$x$ will transition to $\inset$ at the end of round $i+2$ if its $\ell$-value is larger than that of all the agents in $N(x)$.
Furthermore, $x$  will get credit for eliminating the edges incident to $y$ (from the $y$ direction) if its $\ell$ value is larger all the agents in $N(x) \cup N(y) - \{x\}$ in round $i+2$.
Every $u \in N(y)$ is $\unset$ in round $i+2$ by definition. By Lemma \ref{lem:settledown}, $u$ was active in round $i$.
By Lemma \ref{lem:stayElim}, $y$ was also active in round $i$, and therefore the edge $\{u,y\}$ was  active in round $i$.
Thus, if $x$ wins its competition, we can give $1/2$ credit to $x$ for eliminating
the edge $\{u,y\}$ from $\calE_i$. Note that the $1/2$ comes from the fact that the edge $\{u,y\}$ might have been eliminated in a similar manner by from $u$.
By Lemma \ref{lem:fair}, the probability that $x$ dominates the $\ell$-values from
$N(x) \cup N(y) - \{x\}$ is at least $\frac 2 3 (d_x + d_y)$.
Therefore, the total credit due to $x$ for edge $\{x,y\}$ is
$$\frac {d_y}{2} \cdot \frac{2}{3(d_x + d_y)} = \frac{d_y}{3(d_x+d_y)}.$$
Similarly, the credit due to $y$ for the edge $\{x,y\}$ is at least
$d_x/(3(d_x+d_y))$ and the total credit for eliminated edges for the edge
$\{x,y\}$ is $1/3$. Summing  up over all edges in $\hat{E}$,
the expected number of edges from $\calE_i$ that flip from active to eliminated in the next three rounds is at least $|\hat{E}|/3$.
\end{proof}

\section{Self-Jamming Singing Protocol}
\label{sec:SJ}

We now consider the self-jamming model
in which agents hear only the notes from their neighbors that are different than the notes they sing in a round.
The execution of round $(v,i)$ is described in Figure~\ref{fig:synchAlgSJ} and the transition diagram is given in Figure \ref{fig:singTransSJ}. It should be noted that we are
using the notation $s_{v,i}$ to indicate the state of node $v$ 
in iteration $i$ for convenience, node $v$ does not keep track of the iteration $i$.
Much of the protocol is the same as before, and we retain all of the
notation previously introduced, with some modifications as
described below.

\begin{figure}[ht!]
\begin{center}
\noindent\fbox{\begin{minipage}{\textwidth}
\begin{tabbing}
\= ~{\sc Round $(v,i)$ -- (round $i$ for agent $v$)}\\
\> ~~~~~~~\= Compute $\ell_{v,i}$ by flips of a random coin\\
\= ~{\sc Sing during round}\\
\>~~~~~~~\= If $s_{v,i} = \outset$ then do not sing\\
\>~~~~~~~\= If $s_{v,i} = \inset$ then sing even notes 
$0, 2, \ldots,2 \cdot \ell_{v,i}-2$\\
\>~~~~~~~\= If $s_{v,i} = \unset$ then sing odd notes 
$1, 3, \ldots,2 \cdot \ell_{v,i}-1$\\
\=~ {\sc Listen during round}\\
\>~~~~~~~\= If $s_{v,i} = \outset$ then listen for note $0$\\
\>~~~~~~~\= If $s_{v,i} = \inset$ then listen for note $2 \cdot \ell_{v,i}$\\
\>~~~~~~~\= If $s_{v,i} = \unset$ then listen for notes $0$ and  $2 \cdot \ell_{v,i}+1$\\
\=~ {\sc At end of round}\\
\> ~~~~~~~\= If $s_{v,i} = \inset$ then \\
\>\> ~~~~~~~\= If heard note $2 \cdot \ell_{v,i}$ then set $s_{v,i+1} = \unset$\\
\>\> ~~~~~~~\= Else set $s_{v,i+1} = \inset$\\
\> ~~~~~~~\= If $s_{v,i} = \unset$ then \\
\>\> ~~~~~~~\= If heard note $0$ then set $s_{v,i+1} = \outset$\\
\>\> ~~~~~~~\= Else \\
\>\>\> ~~~~~~~\=If heard note $2 \cdot \ell_{v,i}+1$ then set $s_{v,i+1} = \unset$\\
\>\>\> ~~~~~~~\=Else set $s_{v,i+1} = \inset$\\
\> ~~~~~~~\= If $s_{v,i} = \outset$ then \\
\>\> ~~~~~~~\= If heard note $0$ then set $s_{v,i+1} = \outset$\\
\>\> ~~~~~~~\= Else set $s_{v,i+1} = \unset$\\
\end{tabbing}
\end{minipage}}
\end{center}
\caption{Self-jamming singing protocol - round $(v,i)$}
\label{fig:synchAlgSJ}
\end{figure}

\subsection{Analysis for the Synchronous Self-jamming Singing Protocol on Dynamic Networks}


In each round, each agent that is not $\outset$ {\em competes against} its neighbors that have the same state, so the $\unset$ agents compete against its neighbors that are also $\unset$, and an $\inset$ agent
competes its neighbors who are also $\inset$.

We will use the same definition for active and eliminated agents and edges from Definition \ref{def:active} which says that an agent is eliminated if it is $\inNnin$ or has an $\inNnin$ neighbor, otherwise the agent is active. An edge is eliminated if it has at least one eliminated endpoint, and otherwise the edge is active. The set of active edges in round $i$ is denoted by $\calE_i$, and active agents are denoted by $\calV_i$.

The analysis for the self-jamming protocol can be adapted easily to the case of dynamic networks, so we will also assume that the network is dynamic as described in Section r\ref{sec:dynamic-synch}.
We will use the definitions for changed and affected edges from the non-self-jamming protocol given in Definition \ref{def:affected}.

The following two lemmas are adaptations of Lemmas \ref{lem:stayElim} and
\ref{lem:settledown} adapted for the Self-Jamming protocol.

\begin{lemma}
\label{lem:stayElimSJ}
In the Self-Jamming protocol, 
    suppose agent $v$ is incident to an edge that is not affected in the beginning of round $i+1$. Let $u$ be an agent that is either $v$ itself or a neighbor of $v$.     Then  if $u$ was eliminated in round $i$, $u$ will remain eliminated in round $i+1$.
\end{lemma}

\begin{proof}
If $u$ is eliminated in round $i$, then either $u$ is $\inNnin$ or it has a neighbor $u$ that is $\inNnin$.
    By Claim \ref{cl:affected}, the graph induced by all the agents a distance at most $2$ from $u$ remains unchanged at the beginning of round $i+1$
    and execute the protocol correctly at the end of round $i$.
    Consider a neighbor $x$ of $u$ that is $\inNnin$ in round $i$. Since $x$ does not hear any even notes in round $i$, $x$ will continue to remain $\inset$. Moreover, all of $x$'s neighbors (which by assumption execute the protocol correctly at the end of round $i$) are only singing odd notes. Therefore, the will hear note $0$ from $x$ and will remain not $\inset$. Therefore $x$ will remain $\inNnin$, and $u$ will remain eliminated in round $i+1$.
    \end{proof}

\begin{lemma}
\label{lem:settledownSJ}
    Suppose agent $v$ is incident to an edge that is not affected in the beginning of round $i+1$ or $i+2$. If $v$ has a neighbor $u$ that is $\unset$ in round $i+2$, then $u$ was active in round $i$.
\end{lemma}

\begin{proof}
    Suppose a neighbor $u$ of $v$ is eliminated in round $i$.
    We will prove that $u$ cannot be $\unset$ in round $i+2$.
    By Lemma \ref{lem:stayElimSJ}, $u$ will remain eliminated in rounds $i+1$ and $i+2$.
    If $u$ itself is $\inNnin$ in round $i$, then $u$ will remain $\inNnin$ in round $i+2$ and cannot be $\unset$.
    If $u$ has a neighbor that is $\inNnin$ in round $i+1$,
    then $u$ will transition to $\outset$ in round $i+2$.
\end{proof}

The next lemma effectively upper bounds the probability of a conflicted $\inNin$ agent. We leave the proof of Lemma \ref{lem:elig} to the next subsection and proceed to showing how the lemma is used to bound the convergence of the singing self-jamming protocol.

\begin{lemma}
\label{lem:elig}
Consider any configuration of the state of the protocol in round $i$. 
Suppose that $v$ is active in round $i$ ($v \in \calV_i$).
Suppose also that $v$ is incident to an unaffected edge at the beginning of rounds
$i+1$ and $i+2$.
Let $p_v$ be the probability that two rounds later (in round $i+2$) $v$ is $\unset$ and has no neighbors that are $\inset$.
Then the probability that $v$ is eliminated by round $i+2$ ($v \not\in \calV_{i+2}$) is at least $(1-p_v)/2$.
\end{lemma}

The following theorem says that over the course of three rounds, the set of active edges outside of the set of edges affected during those three rounds decreases by a constant factor.

\begin{theorem}
Consider any arbitrary state of the protocol in round $i$. Then
$$\Exp[ |\calE_{i+3}  |] \le \frac 6 7 \cdot  |\calE_i| + \frac 8 7 |  A_{i+3} \cup A_{i+2} \cup A_{i+1}| .$$
\end{theorem}

\begin{proof}
For $v \in \calV_i$, denote the degree of $v$ with respect to the edges in $\calE_i$ as $d_v$. Since the expectation in the statement of the theorem is taken over coin flips during round $i$ or later,  $d_v$ is a fixed value.
Furthermore, if $v$ is incident to an edge that is not affected for the next three rounds, then its neighborhood set in the network does not change during those three rounds.
Note that by Lemma \ref{lem:stayElimSJ}, any edge that is eliminated in round $i$ and is not affected by a change in the network over the next three rounds will remain eliminated. 
Define $\hat{E} = \calE_i - (A_{i+3} \cup A_{i+2} \cup A_{i+1})$
We will show that the expected number of edges in $\calE_i$ that flip from active to eliminated is at least $\frac 1 7 |\hat{E}|$. In the worst case, all the edges in
$A_{i+3} \cup A_{i+2} \cup A_{i+1}$ flip from eliminated to active.
Then the net change in eliminated edges will be:
\begin{align*}
     E[|\calE_{i+3}|] & \le |\calE_i| - \frac 1 7 |\hat{E}| + |A_{i+3} \cup A_{i+2} \cup A_{i+1}|\\
    & = |\calE_i| - \frac 1 7 |\calE_i - (A_{i+3} \cup A_{i+2} \cup A_{i+1})| + |A_{i+3} \cup A_{i+2} \cup A_{i+1}|\\
    & \le \frac 6 7 |\calE_i| + \frac 8 7 |A_{i+3} \cup A_{i+2} \cup A_{i+1}|
\end{align*}

Let $F$ be the set of edges from $\calE_i$ that flip from active to eliminated in the next three rounds.
It remains to show that  $E[|F|] \ge \frac 1 7 |\hat{E}|$.
For each agent $x \in \calV_i$, let $C_x$ be an indicator variable which is $1$ if and only if during round $i+2$, $x$ is $\unset$ and has no neighbors which are $\inset$.
For $\{x,y\} \in \hat{E}$, let $D[x \rightarrow y]$ be an indicator variable which is $1$ if during round $i+2$,
$\ell(x)$ is strictly larger than $\max\{\ell(v) \mid v \in N(x) \cup N(y) - \{x\} \}$. The neighborhood sets of $x$ and $y$ are defined in terms of $\calE_i$.
Note that since every agent generates an $\ell$-value in each round, then $D[x \rightarrow y]$ is always well defined, regardless of the status of agents in
$N(x) \cup N(y) - \{x\}$ during  round $i+2$.
Also note that if $x$  has a neighbor $u$ that is $\unset$ in round $i+2$,
by Lemma \ref{lem:settledownSJ}, then $u$ was active in round $i$.
Therefore, if $x$ is competing against $u$, the edge $\{x,u\} \in \calE_i$
and can potentially be counted as eliminated over the next two rounds.
The same argument holds for neighbors of $y$.

The first bound on the number of edges that drop out of $\calE_i$ comes
from eliminating edges in the next two rounds:
\begin{equation}
    \label{eq:bound1}
    |F| \ge \frac 1 2 \sum_{\{x,y\} \in \calE_i} \left( \Pr[x \not\in \calV_{i+2}] + \Pr[y \not\in \calV_{i+2}] \right)
\end{equation}
We also know that 
\begin{equation}
    \label{eq:bound2}
   |F| \ge \frac 1 2 \sum_{\{x,y\} \in \calE_i} \left(  C_x \cdot D[x \rightarrow y] \cdot d_y  + C_y \cdot D[ y \rightarrow x] \cdot d_x \right).
\end{equation}
This follows from the fact that if $C_x$ occurs, then at the beginning of round $i+2$, agent $x$ is $\unset$ and has no $\inset$ neighbors, and is therefore eligible to transition to $\inset$.
Moreover, $x$ will transition to $\inNnin$ if it's $\ell$-value is strictly larger than all the agents it is competing against (which is a subset of $N(x)$). 
Recall that since $x$ is incident to an edge that is unaffected at the beginning of round $i+2$, $x$ and all of its neighbors follow the protocol correctly
at the end of round $i+1$.
Moreover, we can safely give $x$ sole credit for eliminating all the edges incident to $y$ from the $y$ direction if $\ell(x)$ is  strictly larger than all the agents that are incident to $y$ during round $i+2$ (which is a subset of $N(y)$). The factor of $1/2$ comes from the fact that one edge could possibly be eliminated from either direction, thus we are possibly overcounting by a factor of $2$.

Note that $D[x \rightarrow y]$ depends only on the coin flips that occur during round $i+2$ and $C_x$ depends on the coin flips from rounds $i$ and $i+1$.  Therefore the two variables are independent.
Using Lemma  \ref{lem:win}, we know that
$$\Pr[ D[x \rightarrow y] = 1] \ge \frac 2 3 \cdot \frac{1}{d_x + d_y}.$$
Let $p_x$ denote $\Pr[C_x = 1] = p_x$. Then we can take the expectation for the lower bound
in Equation \ref{eq:bound2} as follows:
\begin{align*}
E[|F|]  \ge & \frac 1 2 \sum_{\{x,y\} \in \hat{E}}  p_x \cdot \frac 2 3 \cdot \frac{d_y}{d_x + d_y}  + p_y \cdot \frac 2 3 \cdot \frac{d_x}{d_x + d_y}  \\
 \ge & \frac 1 3   \sum_{\{x,y\} \in \hat{E}}  \min\{p_x,p_y\}  \\
\end{align*}

By Lemma \ref{lem:elig}, if $\Pr[C_x = 1] = p_x$, then $\Pr[x \not\in \calV_{i+2}] \ge (1 - p_x)/2$,
which we can use to take the expectation for the lower bound
in Equation \ref{eq:bound1} as follows:
\begin{align*}
 \Exp[ |F|]  \ge & \frac 1 2 \sum_{\{x,y\} \in \hat{E}}   \frac{1 - p_x}{2} + \frac{1 - p_y}{2}  \\
 \ge &  \sum_{\{x,y\} \in \hat{E}}  \frac{1 - \min\{p_x,p_y\}}{4} \\
 = & \frac 1 4  |\calE_i| - \sum_{\{x,y\} \in \hat{E}} \frac 1 4 \min\{p_x,p_y\}\\
\end{align*}
The maximum of the two lower bounds  is minimized for $\sum_{\{x,y\} \in \calE_i} \min\{p_x,p_y\} = 3/7\cdot |\calE_i|$, which gives a lower bound
of $|\calE_i| - \Exp[ |\calE_{i+3}|] \ge |\calE_i|/7$.

\end{proof}

\subsubsection{Proof of Lemma \ref{lem:elig}}

Before proving Lemma \ref{lem:elig}, we need a general technical lemma stated below.

\begin{lemma}
\label{lem:X}
Consider any configuration of the state of the protocol during round $i$. Let $X \subseteq \calV_i$.
Let $N(X)$ be the union of the neighbors of $X$ in the network.
Suppose that the network induced by $X \cup N(X)$ does not change at the beginning of round $i+1$, and furthermore, the agents in $X \cup N(X)$ execute the protocol correctly at the end of round $i$.

Let $p_X$ be the probability that in round $i+1$, $X$ has no agents that are $\inset$.
Then the probability that $X$ has at least one agent that is $\inNnin$ in round $i+1$ is at least $(1 - p_X)/2$.
\end{lemma}

\begin{proof}
$p_X$ is defined to be the probability that none of the agents in $X$ is $\inset$ at the end of the round. Therefore,
 the probability that $X$ has at least one $\inset$ agent is $1 - p_X$.
 Let $A$ be the event that  $X$ has at least one $\inNnin$ agent.
 Let $B$ be the event that  $X$ has an $\inset$ agent but no $\inNnin$ agents. 
 $Pr[A] + Pr[B] = 1 - p_X$. We will show that $Pr[A] \ge Pr[B]$,
 which will imply that  $Pr[A] \ge (1 - p_X)/2$.


The following argument depends only on the behavior of $X \cup N(X)$, so we can assume that the network is static.
Let $\cal R$ be the set of random coin flips which results in 
 at least one agent from $X$ that is $\inset$ but no agents that are $\inNnin$, and let $r \in {\cal R}$. 
 Thus, the probability of event $B$ is equal to the probability of selecting a string from ${\cal R}$.
Under random choices $r$, there must be at least one agent
that starts out $\inset$ or $\unset$ that has an $\ell$-value which is equal to the maximum $\ell$-value among the agents it is competing against. Assume the agents have a unique label and select the lexicographically smallest such agent and call it $v$. Suppose that $\ell(v)$ during round $i$ is $L$. One of $v$'s neighbors that it is competing against also had an $\ell$-value of $L$ and all of $v$'s competitors had an $\ell$-value at most $L$.
Now let $S_r$ be the set of outcomes in which all the $\ell$-values are the same, except that $\ell(v) > L$. 
In all the outcomes from $S_r$, $X$ ends up with an agent in $\inNnin$.
The probability of event $A$ is at least the probability of selecting a 
random string from $\bigcup_{r \in {\cal R}} S_r$.
Moreover the probability that $\ell(v) > L$ is equal to the probability that $\ell(v) = L$, so $\Pr[S_r] = \Pr[r]$. 

It remains to establish that 
 if $r, r'\in {\cal R}$ and $r \neq r'$, then $S_r \cap S_{r'} = \emptyset$.
 Let $v$ be the agent chosen for string $r$. The value of $\ell(v)$ under $r$ (which we will denote by $\ell_r(v)$) is determined by $v$'s neighbors that it is competing against. Recall that  $\ell_r(v)$ must be equal to the maximum $\ell$-value of all the neighbors against which it competes. $S_r$ is the set of all strings that can be obtained by starting with $r$ and raising the $\ell$-value for $v$.
 Similarly, let $v'$ be the agent chosen under $r'$.
 
 If the $\ell$-values for $V - \{v,v'\}$ are different under $r$ and $r'$,
 then there is no way to get the same string by changing the $\ell$-values for $v$ and $v'$, which means that $S_r \cap S_{r'} = \emptyset$.
 Therefore, we can assume that the $\ell$-values for $V - \{v,v'\}$ are the same under $r$ and $r'$.
 
 If $v = v'$, then it must be the case that $r = r'$ since the $\ell$-values for $v$ and $v'$ are determined by the $\ell$-value of the neighbors that they are competing against.
 Therefore, we can assume that $v \neq v'$. Without loss of generality, assume that $v$ is the lexicographically smaller of the two agents.
 
 Now suppose that there is a string $s \in S_r \cap S_{r'}$.
 Then $s$ can be reached by starting with $r$ and raising the $\ell$-value of $v$. Also, $s$ can be reached by starting with $r'$ and raising the $\ell$-value  of $v'$. 
 It must be the case then that $\ell_r(v) < \ell_s (v) = \ell_{r'}(v)$
 and $\ell_r(v') = \ell_s(v') > \ell_{r'}(v')$.
 Since $v$ was not chosen under $r'$, it must be the case that  $\ell_{r'}(v)$  is strictly smaller than one of the neighbors it is competing against, whereas under $r$,  $\ell_r(v)$ is at least as large as all of its neighbors that it is competing against. All $\ell$-values are the same under $r$ and $r'$, except for $v$ and $v'$. Therefore, it must be the case  that $v$ and $v'$ are competing against each other and $\ell_{r'}(v) < \ell_{r'}(v')$.
 The three inequalities imply that $\ell_r(v) < \ell_{r}(v')$. This leads to a contradiction because $v$ was chosen under $r$ because its $\ell$-value is at least as large as its neighbors that it is competing against.
\end{proof}

Finally, here is the proof of Lemma \ref{lem:elig}, which is used to show convergence of the singing self-jamming protocol.

\begin{proof}
Define $X$ to be the set that contains $v$ and all of $v$'s neighbors in $\calV_i$.
Since $v$ is incident to an unaffected edge at the beginning of round $i+1$ and $i+2$,
by Claim \ref{cl:affected}, we can assume that the network induced by $X \cup N(X)$ is unchanged at the beginning of rounds $i+1$ and $i+2$. Furthermore all of these agents implement the protocol correctly at the end of rounds $i$ and $i+1$. 
Therefore, we can apply Lemma \ref{lem:X} to $X$ in both rounds.

Let $O_1$ be the event  that none of the agents from $X$  are $\inset$ in round $i+1$, and let $O_2$ be the event that none of the agents from $X$  are $\inset$ at  the end of round $i+1$.
Note that if $O_1$ happens,  then going into round $i+1$, none of $v$'s neighbors are $\inset$ which means that it will transition to $\unset$ or $\inset$ by the end of round $i+1$. Then if $O_2$ also happens (in addition to $O_1$),
then at the end of round $i+1$, $v$ is $\unset$ and none of its neighbors are $\inset$. Therefore $Pr[O_1 \wedge O_2] \ge p_v$.

If $\Pr[O_1] = p_1$, then  by Lemma \ref{lem:X}, the probability that $v \in \barN_{i+1}$ is at least $(1-p_1)/2$. 
Meanwhile if $O_1$ happens then $v$ is still in $N_{i+1}$.
Let ${\cal R}$ be the set of random outcomes for round $i$
which result in event $O_1$, so $\sum_{r \in {\cal R}} p(r) = \Pr[O_1]$.
For each $r \in {\cal R}$, define $p_{2, r} = \Pr[O_2 \mid r]$.
The probability that $O_1$ and $O_2$ both happen is
$$\Pr[O_1 \wedge O_2] = \sum_{r \in {\cal R}} p(r) \cdot \Pr[O_2 \mid r].$$
Under some random choice $r$ in which no agents from $X$ are $\inset$,
it's possible that some of those agents are no longer active because they
now have an $\inNnin$ neighbor. These agents will transition to $\outset$ in the next round and will not effect whether $X$ has an $\inset$ agent in the next round.
We can apply Lemma \ref{lem:X} again to $v$ and its neighbors from $\calV_{i+1}$
to determine that
the probability that $v \not\in \calV_{i+2}$ conditioned on $r$
is at least $(1-p_{2,r})/2$. 
The probability that $v$ is eliminated by round $i+2$ is at least 
$$\frac{1-p_1}{2} + \sum_{r \in {\cal R}} p(r) \left[ \frac{1-p_{2,r}}{2} \right]$$
The first term is a lower bound on the probability that $v$ becomes eliminated at the end of  round $i$ (i.e. $v \not\in \calV_{i+1}$), and the second term is a lower bound on the probability that $v$ is active in round $i+1$ and becomes eliminated at the end of round $i+1$.
Therefore, the probability that $v$ is eliminated by round $i+2$ is at least 
$$\frac{1-p_1}{2} + \frac{p_1 - \Pr[O_1 \wedge O_2]}{2}  \ge \frac{1 - p_v}{2}.$$
 
\end{proof}

\section{The Asynchronous Model}


In our analysis of the asynchronous model, we assume that the communication network is static, and all agents start in the $\unset$ state and the network is static. 
There can be a {\em transmission delay} for a sung note to 
travel from one agent to another.
For example, if there is a constant delay $\delta$ between singing at $u$ and
hearing at $v$, and $u$ sings a note from time $t$ to time $\tp$, 
then that note is heard by $v$ from time $t + \delta$ to time $\tp + \delta$.  
The delay may vary at different times between different adjacent neighbors.

As before, $(v,i)$ is round $i$ for agent $v$.
Let $t_{v,i}$ be the {\em round start time} of $(v,i)$, and thus the start time  
$t_{v,i+1}$ for $(v,i+1)$ is also the {\em round end time} of $(v,i)$.
The {\em round duration} of $(v,i)$ is thus $t_{v,i+1} - t_{v,i}$.
We use $\ell_{v,i}$ to denote the $\ell$-value that $v$ uses during round $(v,i)$, 
i.e., $\ell_{v,i}$ is the value $v$ chooses at the beginning of round $(v,i)$.
As before, $s_{v,i}$ is the state of agent $v$ during round $(v,i)$.
The state of an agent is extended to be defined at a continuous point in time for the asynchronous mode  as follows: If $s_{v,i} = s$ during a round $(v,i)$  then we say $s_{v,t} = s$  for any time $t$ in the middle of round $i$ for  agent $v$: $t_{v,i} \le t < t_{v,i+1}$.

The round duration for an agent can vary over time, 
and the round duration of a given round can depend 
on all previous information available to the agent 
before the beginning of that round, 
which can include all previous 
transmissions sent and received by that agent, 
all $\ell$-values used for previous rounds by that agent,
etc.  The restrictions are that the
round duration is determined before the round begins, 
and thus cannot depend on the $\ell$-value used in that
round or any transmissions received or sent within that round.
In addition, the transmission delay does not
depend on which notes are being sung.

\vspace{.1in}
\noindent{\bf Assumption:}
Let $\Tmin$ be a lower bound on the round duration for any round, 
and let $\Dmax$ be an upper bound on the transmission delay between 
any pair of neighboring agents.
Then, then we assume that the  following invariant holds:
\begin{equation}
    \label{invariant}
2 \cdot \Dmax \le \Tmin.
\end{equation}

We say that a round $(v,j)$ {\bf hears from} another round $(u,i)$, if the notes sung by $u$ in round $(u,i)$ reach agent $v$ by time $t_{v,j+1}$. 
The following lemma is important for establishing that the protocols work correctly and not depend on the protocol.

\begin{lemma}
\label{lem:time}
Consider two rounds $(v,j)$ and $(u,i)$, where $(u,i)$ ends first.
Then they either hear each other, or $(v,j)$ will hear some $(u,k)$, where $k > i$.
\end{lemma}

\begin{proof}
If $(u,i)$ ends before $ t_{v,j} - \Delta_{max}$, then $(v,j)$ will hear some $(u,k)$, where $k > i$.
If $(u,i)$ ends in the period from $ t_{v,j} + \Delta_{max}$ to $t_{v,j+1}$, then they will both hear each other.

Otherwise, $(u,i+1)$ starts some time in the interval $\pm \Delta_{max}$
from $t_{v,j}$. Since the length of both $(u,i+1)$ and $(v,j)$ are both larger than $2 \Delta_{max}$,
$(u,i+1)$ will be singing some time during the interval from $t_{v,j} + \Delta_{max}$ to
$t_{v,j+1} - \Delta_{max}$, which means the signal will be heard by $v$ by time $t_{v,j+1}$.
    \end{proof}

\subsection{Analysis for the Singing Protocol Asynchronous Rounds}

Round $i$ of the singing protocol for agent $v$ is the same as described in Figure~\ref{fig:synchAlg}, which is the same for the synchronous model.

\subsubsection{Correctness analysis}

\begin{lemma}
The states of a pair of neighboring agents are never both \inset.
\end{lemma}

\begin{proof}
    Suppose two adjacent vertices $u$ and $v$ both transition to $\inset$. The rounds in which they transition are $(u,i)$ and $(v,j)$. So $u$ and $v$ are both $\unset$ during those intervals.
Suppose that $(u,i)$ ends before $(v,j)$. By Lemma \ref{lem:time},
either $(u,i)$ and $(v,j)$ can hear each other or $(v,j)$  will hears some round $(u,k)$, such that $k > i$.
    
    If rounds  $(u,i)$ and $(v,j)$ can hear each other, then each  will only transition to $\inset$ of its $\ell$-value is strictly larger than the other's. It cannot be the case that  $\ell_{u,i} > \ell_{v,j}$ and $\ell_{u,i} < \ell_{v,j}$, so they cannot both transition to $\inset$.

    If $(v,j)$  will hears some round $(u,k)$, such that $k > i$, then it will hear not $0$ and will not transition to $\inset$.
\end{proof}

\subsubsection{Convergence time analysis}

The definitions for eliminated and active edges in Definition \ref{def:active} can be extended to a continuous time $t$.
We will denote $\calV_t$ to be the set of active agents at time $t$ and $\calE_t$ to be the set of active edges at time $t$.
We will analyze the protocol over an interval of length $$\Tdiff = 2 \cdot \Dmax + 3 \cdot \Tmax.$$ The interval begins at some time $t$ and ends at time $t' = t + \Tdiff$.
We will lower bound the expected number of active edges from time $t$ that become eliminated by time $t'$.

The following definitions will be used throughout the proof.
The neighborhood of an agent $x \in \calV_t$ denoted $N(x)$ will be be defined in terms of the edges from $\calE_t$. We will use the notation $N(x,y)$ to denote the unions of the neighborhoods of $x$ and $y$:
$$N(x,y) = N(x) \cup N(y),$$
and $d_{xy} = |N(x,y)|$.

$\NR(x, y)$, $\NR(x)$, and $\NR(y)$ are  used to denote the set of rounds for agents in $N(x,y)$, $N(x)$, or $N(y)$ that starts in the interval $(t,t')$:
\begin{align*}
\NR(x, y) & = \{ (u,j): u \in N(x,y) 
\mbox{ and } t < t_{u,j} < \tp \} \\
\NR(x) & = \{ (u,j): u \in N(x) 
\mbox{ and } t < t_{u,j} < \tp \} \\
\NR(y) & = \{ (u,j): u \in N(y) 
\mbox{ and } t < t_{u,j} < \tp \} 
\end{align*}
Let $$\RNmax(x,y) = \frac{\Tdiff}{\Tmin} \cdot d(x,y).$$
Even though $|\NR(x, y)|$ may vary during 
an execution of the protocol depending on the variable
duration of each round (which may depend on the 
$\ell$-values used in the rounds), 
it is always the case that $|\NR(x, y)| \le \RNmax(x,y)$ 
in each execution.
Let $\NL(x,y)$ be a set of $\RNmax(x,y)$ independently
chosen $\ell$-valued random variables.
Thus, it is always possible to assign each round in $\NR(x, y)$
a unique $\ell$-value in $\NL(x,y)$ during an execution of the protocol.  
We let $\ell_{u,j}$ 
denote the $\ell$-value from $\NL(x,y)$ assigned to round $(u,j)$
during an execution of the protocol.

We will define a time $\ta = t +  \Dmax + 2 \cdot\Tmax$ occurring in the interval $(t,t')$.
For each $x \in \calV_t$, we will denote a special round $(x,i_x)$ to be the round for $x$ that spans time $\ta$:
$$t_{x,i_x-1} < \ta \le t_{x,i_x}.$$
$\ell_{x,i_x}$ is the $\ell$-value from $\NL(x,y)$ assigned to rounds $(x,i_x)$.

We will use the  definitions for the following events to bound the probability that the $\ell$-value for $x$ in interval $(x,i_x)$ dominates the $\ell$-values for all of the other intervals in $\NL(x,y)$. Similarly for agent $y$:
\begin{equation*}
\begin{split}
G(L,x \rightarrow y) & = 1 \mbox{ if } \max\{\ell \in \NL(x,y) 
- \{\ell_{x,i_x} \}\}  < L, \\ 
G(L,y \rightarrow x) & = 1 \mbox{ if } \max\{\ell \in \NL(x,y) 
- \{\ell_{y,i_y}\}\}  < L, \\ 
H(L,x) & = 1 \mbox{ if } \ell_{x,i_x} \ge L, \\
H(L,y) & = 1 \mbox{ if } \ell_{y,i_y} \ge L, 
\end{split}
\end{equation*}

The events $F(x)$ and $F(y)$ represent that $x$ or $y$ are eliminated by the end of their designated intervals.
\begin{align}
    F(x ) & = 1 \mbox{ if } \exists (u,j) \in \NR(x) \mbox{ s.t. } s_{u,j} = \inset \mbox{ and } t_{u,j} < t_{x,i_x+1}, \\
F(y) & = 1 \mbox{ if } \exists (u,j) \in \NR(y) \mbox{ s.t. } s_{u,j} = \inset \mbox{ and } t_{u,j} < t_{y,i_y+1}.
\end{align}

Note that the $\ell$-values for rounds that start after time $t$ for 
agents that are in state $\inset$ or $\outset$ at time $t$ 
do not have any further influence in the protocol after time $t$.
Thus, let $\R$ be the set of all rounds that start between $t$ and $\tp$ 
for agents that are not in state $\inset$ or $\outset$ at time $t$:
$$ \R = \left\{ (u,j): u \in \unset~\mbox{at time}~t, ~\mbox{and}~ t < t_{u,j} < \tp \right\}.$$
Note that $\NR(x,y) \subseteq \R$.
The first lemma lower bounds the expected number of edges from $\calE_t$ that are eliminated in the interval $(t,t')$:

\begin{claim}
\label{claim:1}
The expected number of edges from $\calE_t$ eliminated by time 
$\tp$ is at least $1/4$ times
\begin{equation}
\begin{split}
\label{eq:F}
  \sum_{\{x,y\} \in \calE_t} & \Exp[F(x)] + \Exp[F(y)] \\
  + & \Exp\left[(G(L,x \rightarrow y) -F(x))\cdot H(L,x) \cdot d(y)\right] \\
+ \mbox{  } & \Exp\left[(G(L,y \rightarrow x) -F(y))\cdot H(L,y) \cdot d(x)\right],
\end{split}    
\end{equation}
where the expectation is over the choices of the $\ell$-values for rounds in $\R$. 
\end{claim}

\begin{proof}
A credit can be assigned to each edge in $\calE_t$ according to the different ways it is counted as eliminated.
The credits for any edge should add up to at most
one, so that the sum of the credits is a lower bound on the number of edges eliminated
from $\calE_t$. Described below are four ways an edge is counted as eliminated,
so the credit can be set to $1/4$.

Since edge $\{x,y\} \in \calE_t$ is no longer in $\calE_t$ at time $\tp$ if $F(x) = 1$
(since agent $x$ is no longer active by time $t_{x,i_x+1} \le \tp$)
or $F(y) = 1$ (since agent $y$ is no active by time $t_{y,i_y+1} \le \tp$), 
a credit of 1/4 can be given for each of these events.

If $F(x) = 0$, then $x$ is not adjacent to an $\inset$ agent before the end of it's interval at time
$t_{x,i_x+1}$.
Therefore $x$ will be $\unNnin$ during the interval $(x,i_x)$ and will be competing with all its neighbors in the round. $G(L, x \rightarrow y)$ is the event that the $\ell$-value for $x$ during round
$(x,i_x)$ dominates the $\ell$-values in all the rounds for all the neighbors of $x$ and $y$ that start in the whole interval $(t,t')$. So if:
\begin{equation}
\label{eq:fundx}
(G(L, x \rightarrow y) - F(x)) \cdot H(L, x) = 1
\end{equation} 
then the state of $x$
transitions to \inset\ at time $t_{x,i_x+1} \le \tp$, i.e., $x$ is in the $\inset$ state in round $(x,i_x+1)$, which changes the state of $y$ 
from $\unNnin$ to $\unNin$, and eliminates all $d_y$ active edges incident to $y$.
Furthermore, no other agent gets credit for changing the state of $y$ 
from $\unNnin$ to $\unNin$. Note that the left-hand side of Equation~(\ref{eq:fundx}) is less than or equal to zero if it is not equal to one, and thus its expected value is a lower bound on the expected probability of this event.

Thus, agent $x$ can take credit $1/4$ for removing edge $(y, u)$ from $\calE_t$ for each $u \in N(y)$,
for a total credit of $d_y/4$.
(Some other edge $(v, u) \in \calE_t$ may also be given credit $1/4$ for eliminating $(y, u)$
by changing the state of $u$ to eliminated when the state of of $v$ transitions to \inset).

Similarly, if 
\begin{equation}
\label{eq:fundy}
(G(L, y \rightarrow x) - F(y)) \cdot H(L, y) = 1
\end{equation} 
then the state of $y$ gets credit for eliminating $d_x/4$ active edges.
Note that Equation~(\ref{eq:fundx}) and~(\ref{eq:fundy}) are mutually exclusive events.


\end{proof}

The next step is to select a particular value for $L$ to make the probabilities of events $G$ and $H$ favorable for the final lower bound on the number of edges eliminated.

\begin{claim}
\label{cl:lhat}
Let $\Lh = \lceil \log_2 (\RNmax(x,y)) \rceil+1$
    For every $\{x,y\} \in \calE_t$.
    \begin{equation*}
\Exp[G(\Lh,x \rightarrow y)] \ge \frac{1}{3} ~\mbox{  and  }~ \Exp[G(\Lh,y \rightarrow x)]  \ge \frac{1}{3}.
\end{equation*}
Furthermore,
\begin{equation*}
\Exp[H(\Lh,x)] \ge \frac{1}{2 \cdot \RNmax(x,y)} 
~\mbox{ and }~ \Exp[H(\Lh,y)] \ge \frac{1}{2 \cdot \RNmax(x,y)}.
\end{equation*}
The expected values for $G$ and $H$ are over $\ell$-values in $\NL(x,y)$.
\end{claim}

\begin{proof}
The probability that a single $\ell$ is at least $\Lh$ is $(1/2)^{\Lh-1}$.
Therefore $$\Exp[H(\Lh,x)] \ge \frac{1}{2^{\Lh-1}} \ge \frac{1}{2 \cdot \RNmax(x,y)}.$$
The probability that a single $\ell$ is less than $\Lh$ is at least
$1 - (1/2)^{\Lh-1}$. The probability that $R = \RNmax(x,y)$ independent selections for $\ell$ are all less than $\Lh$ is
$$\left( 1 - \frac{1}{2^{\Lh-1}} \right)^R \ge \left( 1 - \frac{1}{2^{\Lh-1}} \right)^{2^{\Lh-1}}$$
Since $\RNmax(x,y) \ge 6$, it follows that $2^{\Lh-1} \ge 8$, and since 
$\left(\frac{7}{8}\right)^8 \ge \frac{1}{3}$ and it follows
that $\Exp[G(\Lh,x \rightarrow y)] \ge \frac{1}{3}$.
\end{proof}

We are now ready to lower bound the expression from Claim \ref{claim:1}.

\begin{claim}
\label{claim:4}
For each $(x,y) \in \calE_t$, 
\begin{equation}
\begin{split}
\label{eq:Fxy}
 & \Exp[F(x )] + \Exp[F(y)] \\
  + & \Exp\left[(G(\Lh,x \rightarrow y) -F(x ))\cdot H(\Lh,x) \cdot d(y)\right] \\
+ \mbox{  } & \Exp\left[(G(\Lh,y \rightarrow x) -F(y ))\cdot H(\Lh,y) \cdot d(x)\right] \\
\ge \mbox{  } & \frac{\Tmin}{7 \cdot \Tdiff},
\end{split}    
\end{equation}
where the expectation is over the $\ell$-values for rounds in $\R$. 
\end{claim} 

\begin{proof}

If either $\Exp[F(x )] \ge \frac{\Tmin}{7 \cdot \Tdiff}$ or 
$\Exp[F(y)]\ge \frac{\Tmin}{7 \cdot \Tdiff}$
then Claim~\ref{claim:4} is true. 
Thus, suppose $\Exp[F(x )] < \frac{\Tmin}{7 \cdot \Tdiff}$ and 
$\Exp[F(y )] < \frac{\Tmin}{7 \cdot \Tdiff}$.

We rewrite
\begin{equation}
\label{eq:GmFx}
\Exp\left[(G(\Lh,x \rightarrow y) -F(x ))\cdot H(\Lh,x) \cdot d(y)\right]
\end{equation}
as
\begin{equation}
\label{eq:GmFxexpanded}
\Exp\left[(G(\Lh,x \rightarrow y) -F(x )) | H(\Lh,x) = 1\right] 
  \cdot \Pr[H(\Lh,x)=1] \cdot d(y).
\end{equation}

For each fixed set of $\ell$-values for rounds $\R - \{(x,i_x)\}$,
and for any positive integer value $L$,
if $F(x )=0$ when $\ell_{x,i_x} = L$ then $F(x )=0$ 
when $\ell_{x,i_x} = L+1$.  Thus,
\begin{equation}
\label{eq:indep}
\Pr[F(x )=1 | H(\Lh,x)=1] \le \Pr[F(x )=1]. 
\end{equation} 
Since $G(\Lh,x \rightarrow y)$ is independent of $H(\Lh,x)$, 
Equation~(\ref{eq:indep})) and Equation~(\ref{eq:GmFxexpanded}) imply that
\begin{equation}
\label{eq:xxx}
\begin{split}
\Exp\left[(G(\Lh,x \rightarrow y) -F(x ))\cdot H(\Lh,x) \cdot d(y)\right] \\
\ge \mbox{  } \Exp[(G(\Lh,x \rightarrow y) -F(x ))] 
  \cdot \Exp[H(\Lh,x)] \cdot d(y).
  \end{split}
\end{equation}
Similar reasoning can be used to show that
\begin{equation}
\label{eq:yyy}
\begin{split}
& \Exp\left[(G(\Lh,y \rightarrow x) -F(y ))\cdot H(\Lh,y) \cdot d(x)\right] \\
\ge \mbox{  } &  \Exp[(G(\Lh,y \rightarrow x) -F(x ))] \cdot \Exp[H(\Lh,y)] \cdot d(x).
  \end{split}
\end{equation}

Since  $\Exp[G(\Lh,x \rightarrow y)] \ge \frac{1}{3}$, 
$\Pr[F(x )] \le \frac{\Tmin}{7 \cdot \Tdiff}$,
and $\frac{\Tmin}{\Tdiff} \le 1/3$ it follows that 
\begin{equation*}
\Exp[(G(\Lh,x \rightarrow y) -F(x ))] \ge \frac{1}{3} - \frac{1}{3 \cdot 7} = \frac{2}{7},   
\end{equation*} 
and thus overall, since 
$\Exp[H(\Lh,x)] \ge \frac{1}{2 \cdot \RNmax(x,y)}$, 
Equation~(\ref{eq:xxx}) is at least
\begin{equation*}
\frac{d(y)}{7 \cdot \RNmax(x,y)}.     
\end{equation*}
By similar reasoning, Equation~(\ref{eq:yyy}) is at least
\begin{equation*}
\frac{d(x)}{7 \cdot \RNmax(x,y)}.     
\end{equation*}
By the properties of $\RNmax(x,y)$, the sum of Equation~(\ref{eq:xxx}) and Equation~(\ref{eq:yyy}) is at least
$$\frac{\Tmin}{7 \cdot \Tdiff}.$$
The proof of Claim \ref{claim:4} follows.
\end{proof}

\begin{theorem}
 \label{lem:asynchmain}
 Let $\Tmax$ be an upper bound on the round duration for any agent,
 let $\Tmin$ be a lower bound on the round duration for any agent,
 and let $\Dmax$ be an upper bound on the transmission delay between 
 any pair of neighboring agents.
 Let $\calE_t$ be the active at time $t$, and $$\euu_t = |\calE_t|.$$
 Then, 
 \begin{equation*}
         \label{eq:asynchmain}
        |\calE_t |- \Exp[|\calE_{t+\Tdiff}|] \ge 
        \frac{|\calE_t|}{28} \cdot \frac{\Tmin}{\Tdiff} .
     \end{equation*}
for all times $t \ge 0$, where $\Tdiff = 2 \cdot \Dmax + 3 \cdot \Tmax$. 
 \end{theorem}

 \begin{proof}

The proof of Lemma~\ref{lem:asynchmain} follows from Claims \ref{claim:1} and \ref{claim:4}.

 \end{proof}

\section{Analysis of Self-Jamming Singing for Asynchronous Rounds}

For the Self-jamming protocol, we need to update what it means for one round to hear another:
we say that a round $(v,j)$ {\bf hears from} another round $(u,i)$, if 
$s_{v,j} = s_{u,i}$, and the notes sung by $u$ in round $(u,i)$ reach agent $v$ by time $t_{v,j+1}$. As before, for round $(v,j)$, let $\hear(v,j)$ denote the set of rounds it hears from.
Define
$$M(v,j) = \max_{(u,i) \in \hear(v,j)} \ell_{u,i}$$

\begin{definition}[Rounds in Collision]
\label{def:collision}
Rounds $(v,j)$ and $(u,j)$ are {\bf in collision} if
\begin{enumerate}
    \item They hear  each other.
    \item $s_{v,j}$ and $s_{u,i}$ are both  $\inset$ or both $\unset$
    \item $\ell_{v,j} = \ell_{u,i}$
    \item  $\ell_{v,j} \ge M(v,j)$
    \item $\ell_{u,i} \ge M(u,i)$.
    \end{enumerate}
\end{definition}

Since Lemma \ref{lem:time} still holds for the self-jamming protocol,
the only way for there to be a conflicted pair of agents is for them
to have rounds that are in collision.
Otherwise, according to Lemma \ref{lem:time}, the one that transitions to $\inset$ first will prevent the other from transitioning to $\inset$.

An agent is said to end a round in state $\inset$ {\bf in a stable way}  if its state is $\inset$ at the end of the round and it is not in collision during that round.
Any agent that ends in state  $\inset$ in a stable way remains $\inNnin$ for the rest of the protocol.
We will call such an agent {\bf stable and $\inset$}.

An agent is {\bf active} if it is not stable and $\inset$ and is not adjacent to agent that is stable and $\inset$. Otherwise the agent is eliminated.

\begin{lemma}
\label{lem:clearSJ}
    Let $R$ be an arbitrary set of intervals that begins between times $t$ and $t'$.
        Let $p$ the probability that none of the rounds in $R$ end in $\inset$.
    Then with probability at least $(1-p)/2$, at least one of the rounds ends in $\inset$ in a stable way.
\end{lemma}

\begin{proof}
$p$ is defined to be the probability that none of the rounds in $R$ is $\inset$ at the end of the round. Therefore,
 the probability that $R$ has at least one rounds that ends in $\inset$  is $1 - p$.
 Let $A$ be the event that  $R$ has at least one round that ends $\inset$ in a stable way.
 Let $B$ be the event that  $R$ has at least one round that ends $\inset$ but none that end $\inset$ in a stable way. If $B$ happens then at least one round ends $\inset$ and has a collision.
 $Pr[A] + Pr[B] = 1 - p$. We will show that $Pr[A] \ge Pr[B]$,
 which will imply that  $Pr[A] \ge (1 - p)/2$.


Let $\cal R$ be the set of random coin flips which results in 
 at least one round from $R$ that ends $\inset$ but no rounds that end $\inset$ in a stable way, and let $r \in {\cal R}$. 
 Thus, the probability of event $B$ is equal to the probability of selecting a string from ${\cal R}$.
  For this proof, we will let $\ell_{v,j}$ by $\ell(v,j)$, since we will use a subscript to denote the random selection from $\cal R$.
Under random choices $r$, there must be at least one round with an $\ell$-value which is equal to the maximum $\ell$-value among rounds that it can hear. Pick the round from $R$ with this property that ends soonest. Call it $(v,j)$.
Note that $\ell_r(v,j) = M_r(v,j)$.
Suppose that $\ell(v,j) = L$. 
Now let $S_r$ be the set of outcomes in which all the $\ell$-values are the same, except that $\ell(v,j) > L$. 
In all the outcomes from $S_r$, $R$ has a round (namely $(v,j)$) which ends $\inset$ in a stable way.
The probability of event $A$ is at least the probability of selecting a 
random string from $\bigcup_{r \in {\cal R}} S_r$.
Moreover the probability that $\ell(v,j) > L$ is equal to the probability that $\ell(v,j) = L$, so $\Pr[S_r] = \Pr[r]$. 

It remains to establish that 
 if $r, r'\in {\cal R}$ and $r \neq r'$, then $S_r \cap S_{r'} = \emptyset$.
 Let $(v,j)$ be the round chosen for string $r$. The value of $\ell_r(v,j)$  is determined by $H(v,j)$ since
 $\ell_r(v,j) = M_r(v,j)$.  $S_r$ is the set of all strings that can be obtained by starting with $r$ and raising the $\ell(v,j)$.
 Similarly, let $(u,i)$ be the agent chosen under $r'$.
 
 If the $\ell$-values for $R - \{(v,j),(u,i)\}$ are different under $r$ and $r'$,
 then there is no way to get the same string by changing the $\ell$-values for $(v,j)$ and $(u,i)$, which means that $S_r \cap S_{r'} = \emptyset$.
 Therefore, we can assume that the $\ell$-values for $R - \{(v,j),(u,i)\}$ are the same under $r$ and $r'$.
 
 
 Now suppose that there is a string $s \in S_r \cap S_{r'}$.
 Then $s$ can be reached by starting with $r$ and raising the $\ell_r(v,j)$. Also, $s$ can be reached by starting with $r'$ and raising the $\ell_{r'}(u,i)$. 
 It must be the case then that $\ell_r(v,j) < \ell_s (v,j) = \ell_{r'}(v,j)$
 and $\ell_r(u,i) = \ell_s(u,i) > \ell_{r'}(u,i)$.
 Since $(v,j)$ was not chosen under $r'$, it must be the case that  $\ell_{r'}(v,j)< M_{r'}(v,j)$. 
 Whereas,  $\ell_{r}(v,j) = M_{r}(v,j)$. 
 If $M_r(v,j) \ge M_{r'}(v,j)$, this contradicts the fact that $\ell_{r'}(v,j) > \ell_{r}(v,j)$.
  All $\ell$-values are the same under $r$ and $r'$, except for $(v,j)$ and $(u,i)$. Therefore, it must be the case that $(u,i) \in H(v,j)$ and
  $\ell_r(u,i)< \ell_{r'}(u,i)$, which contradicts that 
 $\ell_r(u,i) = \ell_s(u,i) > \ell_{r'}(u,i)$.
\end{proof}

Many of the definitions for the Self-jamming protocol are the same as the analysis to the Singing protocol without self-jamming.
We extend the time interval as follows:
 $$\Tdiff = 2 \cdot \Dmax + 6 \cdot \Tmax.$$ 
 The interval begins at some time $t$ and ends at time $t' = t + \Tdiff$.
We will lower bound the expected number of active edges from time $t$ that become eliminated by time $t'$.


The time $\ta = t +  \Dmax + 3 \cdot\Tmax$ is larger than in the analysis for the singing and still is contained in 
 in the interval $(t,t')$.
For each $x \in \calV_t$, the special round for $x$ $(x,i_x)$ is now defined to be the second round after the one that spans $\ta$:
$$t_{x,i_x-3} < \ta \le t_{x,i_x-2}.$$
$\ell_{x,i_x}$ is the $\ell$-value from $\NL(x,y)$ assigned to rounds $(x,i_x)$.
The definitions for 
$\NR(x)$, $\NR(y)$, and $\NR(x,y)$ are the same as in the analysis for the singing protocol.

The definitions for indicator variables $G$ and $H$ are also unchanged.
The are used to bound the probability that the $\ell$-value for $x$ in interval $(x,i_x)$ dominates the $\ell$-values for all of the other rounds in $\NR(x,y)$. Similarly for agent $y$:
\begin{equation*}
\begin{split}
G(L,x \rightarrow y) & = 1 \mbox{ if } \max\{\ell \in \NL(x,y) 
- \{\ell_{x,i_x} \}\}  < L, \\ 
G(L,y \rightarrow x) & = 1 \mbox{ if } \max\{\ell \in \NL(x,y) 
- \{\ell_{y,i_y}\}\}  < L, \\ 
H(L,x) & = 1 \mbox{ if } \ell_{x,i_x} \ge L, \\
H(L,y) & = 1 \mbox{ if } \ell_{y,i_y} \ge L, 
\end{split}
\end{equation*}

The indicator variable $F$ is defined differently for the self-jamming singing protocol.
For agent $x$, we define $\pre(x)$ to be the set of rounds  $(u,j)$ such that
$u$ is a neighbor of $x$ and the round starts after time $t$ and ends before the end of $(x,i_x)$.
$\pre(x)$ also contains the two rounds for $x$ that precede $(x,i_x)$:
$$\pre(x) = \{(x,i_x-2), (x,i_x-1)\} \cup \{(u,j): u \in N(x)~\mbox{and}~t < t_{u,j}~\mbox{and}~ t_{u,j+1} < t_{x,i_x+1}\}$$
We now define indicator variables $F$ and $J$:
\begin{itemize}
    \item $F(x) = 1$ if at least one round  from $\pre(x)$ ends $\inset$ in a stable way.
    \item $J(x) = 1$ if at least one round from $\pre(x)$ ends $\inset$.
\end{itemize}

Note that Lemma \ref{lem:clearSJ} says that $\Pr[J(x) = 1]/2 \le \Pr[F(x) = 1]$.
The definition for $F$ is chosen so that the following fact holds:

\begin{claim}
    \label{cl:stable}
    If $J(x) = 0$, then agent $x$ is in the $\unset$ during round $(x,i_x)$. Moreover, $x$ does not hear from any neighbors that are $\inset$ during round $(x,i_x)$.
\end{claim}

\begin{proof}
    The time $\ta$ is chosen so that for each node $u \in N(x)$, the first round for $u$ in $\pre(x)$ has time to start and finish before round $(x,i_x-2)$
    begins. Moreover, any note sung by $u$ in this first interval will have reached its destination before $(x,i_x-2)$ begins.

Note that if $J(x) = 0$, then all the rounds in $\pre(x)$ end with the agent $\unset$ or $\outset$. Since the first round for each $u \in N(x)$ have finished by the time $(x,i_x-2)$ begins, $x$ will not hear from an $\inset$ neighbor during rounds $(x,i_x-2)$, $(x,i_x-1)$, and $(x,i_x)$.
Since $(x,i_x-2)$ also ends with $x$ not $\inset$, agent $x$ will be $\unset$ or  $\outset$ in round $(x,i_x-1)$. Since $x$ does not hear from an $\inset$ neighbor during this round, $x$ will transition to $\unset$ for round $(x,i_x)$.
\end{proof}

Let $\R$ be the set of all rounds that start between $t$ and $\tp$ 
for agents that are active at time $t$:
$$ \R = \left\{ (u,j): u \in \unset~\mbox{at time}~t, ~\mbox{and}~ t < t_{u,j} < \tp \right\}.$$
Note that $\NR(x,y) \subseteq \R$.
The first lemma lower bounds the expected number of edges from $\calE_t$ that are eliminated in the interval $(t,t')$:

\begin{claim}
\label{claim:1SJ}
The  number of edges from $\calE_t$ eliminated by time 
$\tp$ is at least $1/4$ times
\begin{equation}
\begin{split}
\label{eq:FSJ}
  \sum_{\{x,y\} \in \calE_t} & \Exp[F(x)] + \Exp[F(y)] \\
  + & \Exp\left[(G(L,x \rightarrow y) -J(x))\cdot H(L,x) \cdot d(y)\right] \\
+ \mbox{  } & \Exp\left[(G(L,y \rightarrow x) -J(x))\cdot H(L,y) \cdot d(x)\right],
\end{split}    
\end{equation}
where the expectation is over the choices of the $\ell$-values for rounds in $\R$. 
\end{claim}

\begin{proof}
A credit can be assigned to each edge in $\calE_t$ according to the different ways it is counted as eliminated.
The credits for any edge should add up to at most
one, so that the sum of the credits is a lower bound on the number of edges eliminated
from $\calE_t$. Described below are four ways an edge is counted as eliminated,
so the credit can be set to $1/4$.

If $F(x) = 1$, then $x$ has a neighbor that is $\inset$ and stable
or $x$ itself is $\inset$ and stable. Therefore.
the edge $\{x,y\}$ will be eliminated by time $t'$.
Similarly for $F(y) = 1$. Therefore, 
a credit of 1/4 can be given for each of these events.

If $J(x) = 0$, then by Claim \ref{cl:stable}, $x$ is $\unset$ in round $(x,i_x)$.
Moreover, none of the agents in rounds from $\hear(x,i_x)$ are $\inset$.
Therefore, $x$ will transition to $\inset$ in a stable way, if its $\ell$-value
is larger than all of the rounds in $\NR(x)  - \{(x,i_x)\}$. 
Moreover $x$ will get sole credit for eliminating all the active edges incident to $y$
if its $\ell$-value is larger than all the rounds in $\NL(x,y) - \{(x,i_x)\}$.
Therefore:
\begin{equation}
\label{eq:fundxSJ}
(G(L, x \rightarrow y) - J(x)) \cdot H(L, x) = 1
\end{equation} 
then the state of $x$
transitions to \inset\ in a stable way at time $t_{x,i_x+1} \le \tp$.
This changes the state of $y$ 
from $\unNnin$ to $\unNin$, and eliminates all $d_y$ active edges incident to $y$.
Furthermore, no other agent gets credit for changing the state of $y$ 
from $\unNnin$ to $\unNin$. Note that the left-hand side of Equation~(\ref{eq:fundxSJ}) is less than or equal to zero if it is not equal to one, and thus its expected value is a lower bound on the expected probability of this event.

Thus, agent $x$ can take credit $1/4$ for removing edge $\{y,u\}$ from $\calE_t$ for each $u \in N(y)$,
for a total credit of $d_y/2$. Note that some other edge $\{v,u\}$ may also get credit $1/4$ for eliminating $\{y,u\}$ if $v$ transitions to $\inset$).
Similarly, if 
\begin{equation}
\label{eq:fundySJ}
(G(L, y \rightarrow x) - J(y)) \cdot H(L, y) = 1
\end{equation} 
then the state of $y$ gets credit for eliminating $d_x/4$ active edges.
Note that Equation~(\ref{eq:fundxSJ}) and~(\ref{eq:fundySJ}) are mutually exclusive events.


\end{proof}

$\Lh$ is chosen again to be $\lceil \log_2 (\RNmax(x,y)) \rceil+1$.  
We will use the bound from Claim \ref{cl:lhat} that says
\begin{equation*}
\Exp[G(\Lh,x \rightarrow y)] \ge \frac{1}{3} \mbox{  and  } \Exp[G(\Lh,y \rightarrow x)]  \ge \frac{1}{3}.
\end{equation*}
and
\begin{equation*}
\Exp[H(\Lh,x)] \ge \frac{1}{2 \cdot \RNmax(x,y)} 
\mbox{ and } \Exp[H(\Lh,y)] \ge \frac{1}{2 \cdot \RNmax(x,y)}.
\end{equation*}
The expected values for $G$ and $H$ are over $\ell$-values in $\NL(x,y)$. 

\begin{claim}
\label{claim:4SJ}
For each $(x,y) \in \calE_t$, 
\begin{equation}
\begin{split}
\label{eq:FxySJ}
 & \Exp[F(x )] + \Exp[F(y)] \\
  + & \Exp\left[(G(\Lh,x \rightarrow y) -J(x ))\cdot H(\Lh,x) \cdot d(y)\right] \\
+ \mbox{  } & \Exp\left[(G(\Lh,y \rightarrow x) -J(y ))\cdot H(\Lh,y) \cdot d(x)\right] \\
\ge \mbox{  } & \frac{\Tmin}{14 \cdot \Tdiff},
\end{split}    
\end{equation}
where the expectation is over the $\ell$-values for rounds in $\R$. 
\end{claim} 

\begin{proof}

If either $\Exp[F(x )] \ge \frac{\Tmin}{14 \cdot \Tdiff}$ or 
$\Exp[F(y)]\ge \frac{\Tmin}{14 \cdot \Tdiff}$
then Claim~\ref{claim:4} is true. 
Thus, suppose $\Exp[F(x )] < \frac{\Tmin}{14 \cdot \Tdiff}$ and 
$\Exp[F(y )] < \frac{\Tmin}{14 \cdot \Tdiff}$.

Lemma \ref{lem:clearSJ} says that $\Pr[J(x) = 1]/2 \le \Pr[F(x) = 1]$.
Therefore,  $\Exp[J(x )] < \frac{\Tmin}{7 \cdot \Tdiff}$ and 
$\Exp[J(y )] < \frac{\Tmin}{7 \cdot \Tdiff}$.
We rewrite
\begin{equation}
\label{eq:GmFxSJ}
\Exp\left[(G(\Lh,x \rightarrow y) -J(x ))\cdot H(\Lh,x) \cdot d(y)\right]
\end{equation}
as
\begin{equation}
\label{eq:GmFxexpandedSJ}
\Exp\left[(G(\Lh,x \rightarrow y) -J(x )) | H(\Lh,x) = 1\right] 
  \cdot \Pr[H(\Lh,x)=1] \cdot d(y).
\end{equation}

For each fixed set of $\ell$-values for rounds $\R - \{(x,i_x)\}$,
and for any positive integer value $L$,
if $J(x )=0$ when $\ell_{x,i_x} = L$ then $J(x )=0$ 
when $\ell_{x,i_x} = L+1$.  Thus,
\begin{equation}
\label{eq:indepSJ}
\Pr[J(x )=1 | H(\Lh,x)=1] \le \Pr[J(x )=1]. 
\end{equation} 
Since $G(\Lh,x \rightarrow y)$ is independent of $H(\Lh,x)$, 
Equation~(\ref{eq:indepSJ})) and Equation~(\ref{eq:GmFxexpandedSJ}) imply that
\begin{equation}
\label{eq:xxxSJ}
\begin{split}
\Exp\left[(G(\Lh,x \rightarrow y) -J(x ))\cdot H(\Lh,x) \cdot d(y)\right] \\
\ge \mbox{  } \Exp[(G(\Lh,x \rightarrow y) -J(x ))] 
  \cdot \Exp[H(\Lh,x)] \cdot d(y).
  \end{split}
\end{equation}
Similar reasoning can be used to show that
\begin{equation}
\label{eq:yyySJ}
\begin{split}
& \Exp\left[(G(\Lh,y \rightarrow x) -J(y ))\cdot H(\Lh,y) \cdot d(x)\right] \\
\ge \mbox{  } &  \Exp[(G(\Lh,y \rightarrow x) -J(x ))] \cdot \Exp[H(\Lh,y)] \cdot d(x).
  \end{split}
\end{equation}

Since  $\Exp[G(\Lh,x \rightarrow y)] \ge \frac{1}{3}$, 
$\Pr[J(x )] \le \frac{\Tmin}{7 \cdot \Tdiff}$,
and $\frac{\Tmin}{\Tdiff} \le 1/3$ it follows that 
\begin{equation*}
\Exp[(G(\Lh,x \rightarrow y) -J(x ))] \ge \frac{1}{3} - \frac{1}{3 \cdot 7} = \frac{2}{7},   
\end{equation*} 
and thus overall, since 
$\Exp[H(\Lh,x)] \ge \frac{1}{2 \cdot \RNmax(x,y)}$, 
Equation~(\ref{eq:xxxSJ}) is at least
\begin{equation*}
\frac{d(y)}{7 \cdot \RNmax(x,y)}.     
\end{equation*}
By similar reasoning, Equation~(\ref{eq:yyySJ}) is at least
\begin{equation*}
\frac{d(x)}{7 \cdot \RNmax(x,y)}.     
\end{equation*}
By the properties of $\RNmax(x,y)$, the sum of Equation~(\ref{eq:xxx}) and Equation~(\ref{eq:yyySJ}) is at least
$$\frac{\Tmin}{7 \cdot \Tdiff}.$$
The proof of Claim \ref{claim:4SJ} follows.
\end{proof}

\begin{theorem}
 \label{lem:asynchmainSJ}
 Let $\Tmax$ be an upper bound on the round duration for any agent,
 let $\Tmin$ be a lower bound on the round duration for any agent,
 and let $\Dmax$ be an upper bound on the transmission delay between 
 any pair of neighboring agents.
 Let $\calE_t$ be the active at time $t$.
 Then, 
 \begin{equation*}
         \label{eq:asynchmainSJ}
        |\calE_t |- \Exp[|\calE_{t+\Tdiff}|] \ge 
        \frac{|\calE_t|}{98} \cdot \frac{\Tmin}{\Tdiff} .
     \end{equation*}
for all times $t \ge 0$, where $\Tdiff = 2 \cdot \Dmax + 6 \cdot \Tmax$. 
 \end{theorem}

 \begin{proof}

The proof of Lemma~\ref{lem:asynchmainSJ} follows from Claims \ref{claim:1SJ} and \ref{claim:4SJ}.

 \end{proof}

\section{Coin flipping process}
\label{sec:coin}

\begin{lemma}
\label{lem:win}
Consider the following process: we have $d \ge 1$ agents which all start at the same time. In each step, each agent independently flips a bit that is $1$ with  probability $p$ and is $0$ with probability $1-p$, and the agent goes to the next step if the outcome is $1$, and the agent terminates at the step if the outcome is $0$. 

Let $y_d$ be the probability that the first agent executes strictly more steps than any of the other $d-1$ agents before terminating. Then,
\begin{equation*}
    \label{maineqn}
    y_d \ge \frac {2 \cdot p} {(1+p) \cdot d}.
\end{equation*}

\end{lemma}

\begin{proof}
Fix $p$, $0 < p < 1$.  For any $d \ge 1$, let $x_d$ be the probability that the first of $d$ randomly chosen infinite-length binary numbers is larger than all of the other $d-1$ numbers, where each bit of each of the $d$ random binary numbers is chosen to be $1$ with probability $p$ and $0$ with probability $1-p$.  Then, it is clear that $x_d = 1/d$.  

For $d=1$, $x_d = 1$. For $d \ge 2$, the value of $x_d$ can be written as follows.
\begin{equation*}
    x_d = p \cdot \sum_{i=0}^{d-2} \binom{d-1}{i}\cdot p^i \cdot (1-p)^{d-1-i} \cdot x_{i+1} + \left( (1-p)^d + p^d \right)
    \cdot x_d.
\end{equation*}
This can be seen by considering the binary numbers one bit at a time.
The binary number of the first agent is largest if its first bit is 1 (with probability $p$) and, for some $i$ between $0$ and $d-2$, the first bit of the binary number of $i$ of the other $d-1$ agents is 1 and the remaining bits of the binary number of the first agent is larger than the remaining bits of the binary numbers of the other $i$ agents (the sum over $i=0$ through $d-2$ of the probability that $i$ of the other $d-1$ agents first bit are 1 times $x_{i+1}$), or all first bits of the $d$ agents are 0 (with probability $(1-p)^d$) or all first bits are 1 (with probability $p^d$) and the remaining bits of the binary number of the first agent is larger than the remaining bits of the binary numbers of all $d-1$ of the other agents (times $x_d$).

Similarly, for $d=1$, there is a single agent that participates in the first step, which is strictly more steps than any other agent, and thus $y_d = 1$. For $d \ge 2$, the value of $y_d$ can be written as follows.
\begin{equation*}
    y_d = p \cdot \sum_{i=0}^{d-2} \binom{d-1}{i}\cdot p^i \cdot (1-p)^{d-1-i} \cdot y_{i+1} +  p^d \cdot y_d.
\end{equation*}
Define function $\phi_d$ with input $z=\{z_1, z_2, \ldots, z_{d-1} \}$ as follows
\begin{equation*}
\label{phidefeqn}
    \phi_d(z) = p \cdot \sum_{i=1}^{d-2} \binom{d-1}{i}\cdot p^i \cdot (1-p)^{d-1-i} \cdot z_{i+1}.
\end{equation*}
Then,
\begin{equation*}
\label{xeqn}
    x_d = \phi_d(x) + p \cdot (1-p)^{d-1} \cdot x_1  + \left( (1-p)^d + p^d \right)
    \cdot x_d.
\end{equation*}
and \begin{equation*}
\label{yeqn}
    y_d = \phi_d(y) + p \cdot (1-p)^{d-1} \cdot y_1 + p^d \cdot y_d.
\end{equation*}
Note that since $x_1 = 1$ and $y_1 = 1$, Equations~(\ref{xeqn}) and (\ref{yeqn}) simplify to: 
\begin{equation}
\label{xeqnf}
     (1- p^d) \cdot x_d = \phi_d(x) + p \cdot (1-p)^{d-1} + (1-p)^d \cdot x_d.
\end{equation}
and \begin{equation}
\label{yeqnf}
    (1- p^d) \cdot y_d = \phi_d(y) + p \cdot (1-p)^{d-1}.
\end{equation}

Define
\begin{equation}
\label{alphadef}
  \alpha = \frac{p}{x_2 + p \cdot (1-x_2)} = \frac{2 \cdot p}{1+p}.  
\end{equation}  
The proof is by induction on $d$. Note that $\alpha \le 1$, and thus $y_1 \ge \alpha \cdot x_1$, establishing
the base case $d=1$.

For $d \ge 2$, assume for the induction hypothesis that $y_i \ge \alpha \cdot x_i$ for all $i = 1,\ldots, d-1$.  
We need to show that 
$y_d \ge \alpha \cdot x_d$.  By the induction hypothesis and Equation~(\ref{phidefeqn}), it is clear that 
\begin{equation}
\label{phieqn}
\phi_d(y) \ge \alpha \cdot \phi_d(x).
\end{equation}
From Equations~(\ref{xeqnf}), (\ref{yeqnf}), (\ref{phieqn}), $y_d \ge \alpha \cdot x_d$ if 
\begin{equation} 
\label{techeqn}
p \cdot (1-p)^{d-1} \ge \alpha \left( p \cdot (1-p)^{d-1}  + (1-p)^d \cdot x_d \right).
\end{equation}
Equation~(\ref{techeqn}) can be simplified to
\begin{equation} 
\label{techeqns}
\alpha \le \frac{p}{x_d + p \cdot (1 - x_d)}.
\end{equation}
Since $x_2 \ge x_d$ for all $d \ge 2$, the proof of 
the lemma follows since Equation~(\ref{techeqns}) 
is satisfied for all $d \ge 2$
with the value of $\alpha$ defined in Equation~(\ref{alphadef}).

\end{proof}

\begin{lemma}
\label{lem:fair}
  When $p=1/2$, $y_d \ge \frac{2}{3 \cdot d}$   
\end{lemma}

Lemma~\ref{lem:win} can be generalized to the case when the bits for
the different agents have different biases.  For example, for all $j = 1, \ldots, d$, agent $j$ flips a bit that is $1$ with probability $p_j$, where $p_j$ and $p_{j'}$ may not be equal for $j \not= j'$.  

In this case, define 
\begin{equation*}  
\bar{x}^j_2 = \min_{j' \not= j} \left\{  
\frac{p_j \cdot (1-p_{j'})}{p_j \cdot (1-p_{j'}) + p_{j'} \cdot (1-p_j)} \right\}. 
\end{equation*}
Then, the proof of Lemma~\ref{lem:win} can be used to prove that
\begin{equation*}
    y^j_d \ge \alpha_j \cdot x^j_d,
\end{equation*}
where $y^j_d$ and $x^j_d$ are defined analogously to $y_d$ and $x_d$ for agent $j$,
and where 
\begin{equation*}
    \alpha_j = \frac{p_j}{\bar{x}^j_2 + p_j \cdot (1-\bar{x}^j_2)} 
\end{equation*}

\bibliography{main}

\end{document}